\pdfoutput=1
\documentclass[11pt,a4paper]{article}
\usepackage[utf8]{inputenc}
\usepackage[T1]{fontenc}
\usepackage{amsmath,amsthm}
\usepackage{XCharter}
\usepackage[xcharter,vvarbb]{newtxmath}
\usepackage[margin=2.2cm]{geometry}
\usepackage{microtype}
\usepackage{xcolor}
\usepackage{enumitem}
\usepackage{booktabs}
\usepackage{tikz}
\usetikzlibrary{arrows.meta,positioning,calc,decorations.pathreplacing,shapes.misc}
\usepackage{titlesec}
\usepackage{needspace}
\usepackage{tcolorbox}
\tcbuselibrary{skins,breakable}
\usepackage{caption}
\usepackage[hidelinks]{hyperref}
\usepackage[nameinlink,noabbrev,capitalise]{cleveref}
\hypersetup{pdftitle={When Pebbles Are Not Atomic: A Polynomial Separation Between Pure and Entangled Pebbling},pdfauthor={Piotr Marszalik}}

\definecolor{ask}{RGB}{31,94,140}
\definecolor{buy}{RGB}{214,125,32}

\newcommand{\Out}{\mathrm{Out}}
\newcommand{\In}{\mathrm{In}}
\newcommand{\TMC}{\mathrm{TMC}}
\newcommand{\Par}{\operatorname{par}}
\newcommand{\cl}{\operatorname{cl}}
\newcommand{\wt}{\operatorname{wt}}
\newcommand{\avail}{\operatorname{avail}}
\newcommand{\indeg}{\operatorname{indeg}}
\newcommand{\Hk}{\mathcal H_K}
\newcommand{\Gk}{\mathcal G_K}
\newcommand{\Qk}{\mathcal Q_K}
\newcommand{\Bk}{\mathcal B_K}
\newcommand{\Sx}{S_\oplus}
\newcommand{\Cans}{C^{\mathrm{ans}}}
\newcommand{\Cret}{C^{\mathrm{ret}}}
\newcommand{\CPC}{\mathrm{CPC}}
\newcommand{\E}{\mathbb E}
\newcommand{\certmark}[1]{\tikz[baseline=(certnumber.base)]{\node[draw,circle,inner sep=.7pt,line width=.4pt,font=\scriptsize] (certnumber) {#1};}}

\newtheorem{theorem}{Theorem}
\newtheorem{lemma}[theorem]{Lemma}
\newtheorem{proposition}[theorem]{Proposition}
\newtheorem{corollary}[theorem]{Corollary}
\theoremstyle{definition}
\newtheorem{definition}[theorem]{Definition}
\theoremstyle{remark}
\newtheorem{remark}[theorem]{Remark}

\titleformat{\section}{\large\bfseries}{\thesection.}{0.6em}{}
\titlespacing*{\section}{0pt}{1.6em}{0.4em}
\titleformat{\subsection}{\normalsize\bfseries}{\thesubsection}{0.6em}{}

\newtcolorbox{intuition}[1][Intuition.]{%
  enhanced, breakable, parbox=false, frame hidden, sharp corners,
  colback=ask!5, borderline west={1.6pt}{0pt}{ask!70},
  left=9pt, right=7pt, top=5pt, bottom=5pt,
  before skip=8pt, after skip=10pt, fontupper=\small,
  before upper={\textbf{\color{ask!85!black}#1}\enspace}}

\begin{document}

\begin{center}
  {\LARGE\bfseries When Pebbles Are Not Atomic}\par\vspace{4pt}
  {\large A Polynomial Separation Between Pure and Entangled Pebbling}\par\vspace{10pt}
  {\normalsize Piotr Marszalik}\par\vspace{2pt}
  {\small Institute of Science and Technology Austria}\par\vspace{2pt}
  {\small University of Warsaw}\par\vspace{2pt}
  {\small October 2026}
\end{center}
\vspace{0.4em}

\begin{abstract}
\noindent
Graph pebbling is the model in which the memory cost of memory-hard
functions and proofs of space is measured, and pebbling lower bounds are
the main way to argue that such constructions are secure. It has long
seemed that the model is faithful to reality. It is known that it is not
exact: an adversary who stores encoded information instead of whole values
can do better, but so far only by small factors: a factor of $1.5$ in the
work that introduced entangled pebbling (Alwen, Chen, Kamath, Kolmogorov,
Pietrzak and Tessaro) and, recently, a
factor of about $\log n/\log\log n$ (de~Rezende, Engstr\"om and Reyzin).
One could still hope that pebbling is right up to such small factors.

In \emph{pure} pebbling, memory consists of whole labels; \emph{entangled}
pebbling also allows encoded state, such as XORs of labels. We show that
the two models can differ polynomially. We study the setting in which
proofs of space and data-dependent memory-hard functions are used: a state
is prepared in advance, a random challenge has to be answered from it, and
memory is paid for as long as it is held. Our graphs are explicit and have
indegree two, and on them a pure adversary can do no better, up to a
constant factor, than storing every possible answer in advance. Still, an
adversary who keeps XORs instead of labels pays polynomially less. This
ends the hope that pure pebbling is a faithful model of memory up to small
factors.

The heart of the paper is a simple combinatorial game, together with a
special class of graphs found by Alon, Moitra and Sudakov.
\end{abstract}

% ======================================================================
\section{Introduction}
\label{sec:intro}

Saving memory has a price: a value that is not kept must be recomputed if
it is needed again. \emph{Memory-hard
functions}~\cite{AlwenSerbinenko2015,Scrypt2017} are built to make this
price high: evaluating them with little memory should cost a lot of extra
computation. In password hashing, for example, this is what makes guessing
many passwords in parallel on custom hardware expensive, since every guess
needs its own memory. What matters is not only how much memory is used,
but also for how long it stays occupied.

Graph pebbling turns this tension into a combinatorial problem. The
vertices of a directed acyclic graph (DAG) stand for intermediate values,
each computed by a hash function from a few earlier ones, and a pebble on a
vertex records that its value is currently stored. This is the standard
tool for reasoning about the memory cost of memory-hard functions and of
\emph{proofs of space}~\cite{DFKP2015}, in which a party stores data and
later answers challenges about it.

The model has earned its place. For the evaluation of a data-independent
memory-hard function from scratch, a lower bound on pebbling can be turned
into a lower bound against every algorithm, at least in an idealized world where the hash
function is a random oracle~\cite{AlwenSerbinenko2015}. So it is tempting to read a
pebbling bound simply as the truth about memory.

But every memory lower bound is a statement about a model of memory, and
pebbling makes a specific choice: stored information consists of
\emph{complete} vertex values. A real algorithm need not store information
this way. It may keep a function of several values and combine it later
with values it recomputes. The simplest example is XOR: instead of keeping
two labels, keep their XOR; once either label is available again, the other
one comes back for free. Entangled pebbling~\cite{AlwenEtAl2016} formalizes
such encoded states.

How much can such tricks gain? The question is as old as entangled pebbling
itself. Alwen et al.~\cite{AlwenEtAl2016} introduced the model, showed
that storing something other than labels can provably help, but only by a
factor of $1.5$, and proved their lower bound for scrypt in the entangled
model, conjecturing that it extends to arbitrary algorithms. Recently,
de~Rezende, Engstr\"om and Reyzin~\cite{DER2026} showed that XORs beat pure
pebbling for families of proofs of space and memory-hard functions, by an
$\Omega(\log n/\log\log n)$ factor for an indegree-two data-dependent
memory-hard function. Both gaps are small. One could still hope that pure
pebbling is right up to such small factors, and that pebbling bounds, which
are much easier to prove, can be trusted as they are.

\paragraph{Our result.} The gap can be polynomial. We give an explicit
family of DAGs of indegree two on which a pure player can do no better, up
to a constant factor, than caching all possible answers, while a player who
stores XORs of labels is polynomially cheaper. The setting is the one in
which proofs of space and data-dependent memory-hard functions are actually
used: data is prepared in advance, and then a fresh random challenge has to
be answered from it, with memory paid for in every round it is held.

\begin{theorem}[Main result, informal]
\label{thm:intro}
For infinitely many $N$ there is an explicit DAG on $N$ vertices, of
maximum indegree two, and a set of possible challenges, such that the
following holds for a uniformly random challenge answered from a state
prepared in advance (preparing the state is not charged). The best pure
strategy has expected cumulative cost $\Theta(N/\log N)$, which matches,
up to a constant factor, the cost of caching all possible answers. An
explicit XOR strategy answers the challenge and restores its prepared state
at expected cumulative cost $O(N^{1-\eta}(\log N)^{\eta})$, for a constant
$\eta>0$. The gap persists when many challenges are answered one after
another and the preparation is paid for.
\end{theorem}

The precise statement is \cref{thm:main}, and \cref{thm:repeated} is its
version with repeated challenges and paid preparation. The constant $\eta$
is small (we get $\eta=0.029$), but it is a polynomial
gap, and it comes from a very restricted adversary. Its XORs are all of the
same kind, they are computed once and never changed, and every move it
makes is allowed in the entangled pebbling model of Alwen et
al.~\cite{AlwenEtAl2016}. Treating stored values as atomic is therefore not
a harmless simplification.

\paragraph{Where the gap comes from.} The whole separation lives in a
small combinatorial game, which we present first (\cref{sec:game}) and
around which the rest of the paper is built. Picture a table of labels. A
question asks for a whole column, but labels can only be bought a whole row
at a time. A player who stores single labels is in trouble: a stored label
helps only with the questions that happen to need it, and we show that
nothing he stores saves him anything on average. A player who stores XORs
of labels can do much better. Each stored XOR mixes labels from many rows.
When a column is asked for, the player buys just a few rows, and for every
label of the column there is a stored XOR whose other labels all lie in
these rows; XORing them out gives that label back. Different labels of the
column come from different XORs, and the same XOR serves many columns,
giving back a different label each time. To make
this work we need graphs that are nearly complete and yet split into large
induced matchings. Such graphs exist, by a beautiful construction of Alon,
Moitra and Sudakov~\cite{AMS2013}, and they are what makes the gap
polynomial.

The rest of the paper turns the game into an honest computation. First we
build a DAG in which the rules of the game come out of the graph itself
(\cref{sec:dag}); then we bring its indegree down to two without losing the
gap (\cref{sec:degree-two}); finally we answer many challenges in a row
(\cref{sec:repeated}).

\paragraph{Why one challenge from a prepared state.} A prover in a proof of
space stores data once and then answers random challenges from that state;
an evaluator of a data-dependent memory-hard function repeatedly needs the
label of a vertex selected by labels computed just before. In both settings
the recurring cost is not that of building a memory state but that of
answering a request from it. We therefore study one fresh random request
served from a state fixed before the request is known: preparing the state
is not charged, but every unit of memory held while answering is. We hold
the XOR player to a stricter standard: he must also restore his state, so
that he is ready for the next request.

\paragraph{Related work.} Entangled pebbling was introduced by Alwen, Chen,
Kamath, Kolmogorov, Pietrzak and Tessaro~\cite{AlwenEtAl2016}. Pebbling
lower bounds for evaluating a data-independent memory-hard function from
scratch transfer to arbitrary algorithms in the parallel random oracle
model~\cite{AlwenSerbinenko2015}; for data-dependent ones, the connection is
more delicate~\cite{BlockiHolman2022,BlockiHolman2025}. De~Rezende,
Engstr\"om and Reyzin~\cite{DER2026} give an $\Omega(\log n/\log\log n)$
separation for an indegree-two data-dependent memory-hard function; in the
repeated game on our construction the gap is $\widetilde\Omega(n^{\eta})$
(\cref{cor:repeated-n}). For comparison, every constant-indegree DAG on
$M$ vertices has cumulative pebbling cost
$O(M^2\log\log M/\log M)$~\cite{AlwenBlocki2016}, while depth-robust
constructions reach $\Omega(M^2/\log M)$~\cite{ABP2017}. We prove no lower
bound against general random-oracle algorithms; our result compares the
two pebbling models.

\paragraph{Organization.} \Cref{sec:game} is the memory game and its
separation. \Cref{sec:models} fixes the pebbling model and collects the
tools (forced ancestors, connectors). \Cref{sec:dag} plays the game on a
high-indegree DAG, \cref{sec:degree-two} on an indegree-two DAG, and
\cref{sec:repeated} repeats it. \Cref{sec:conclusion} concludes with open
questions. \Cref{app:ams} derives the parameters of the auxiliary graphs, and
\cref{app:dmhf} writes out the function behind the repeated game.

% ======================================================================
\section{The memory game}
\label{sec:game}

This section is self-contained: no DAGs, no hash functions, no pebbles.
Everything that follows is a way of making this game real.

\subsection{The rules}

Fix a graph $G=([K],E)$, simple and undirected, with degrees $d_u$ and
average degree
\[
  d \;=\; \frac1K\sum_{u} d_u \;=\; \frac{2|E|}{K}.
\]
To each edge $\{u,v\}\in E$ we assign two \emph{labels}, $x_{u,v}$ and
$x_{v,u}$ (as if we made $G$ directed by adding both directions). Labels
are independent uniformly random $w$-bit strings, so knowing some labels
tells you nothing about the others. One label takes one unit of memory.
For a vertex $v$ let
\[
  \Out(v) = \{x_{v,u} : u\in N(v)\},
  \qquad
  \In(v) = \{x_{u,v} : u\in N(v)\}
\]
be everything that points out of $v$ and everything that points into $v$.
It helps to picture the labels as the entries of the adjacency matrix:
$x_{u,v}$ sits at position $(u,v)$, $\Out(u)$ is row $u$, and $\In(v)$ is
column $v$ (entries of non-edges are empty; see \cref{fig:matrix}).

\begin{figure}[htbp]
\centering
\begin{tikzpicture}[scale=0.62, every node/.style={font=\small}]
  \fill[buy!22, rounded corners=5pt] (0.1,3.3) rectangle (5.9,3.9);
  \fill[ask!15, rounded corners=5pt] (3.9,0.1) rectangle (4.5,5.9);
  \draw[buy!85!black, rounded corners=5pt] (0.1,3.3) rectangle (5.9,3.9);
  \draw[ask!85!black, rounded corners=5pt] (3.9,0.1) rectangle (4.5,5.9);
  \draw[thick] (0,0) -- (-0.2,0) -- (-0.2,6) -- (0,6);
  \draw[thick] (6,0) -- (6.2,0) -- (6.2,6) -- (6,6);
  \fill (4.2,3.6) circle (2.4pt);
  \draw[->, thin, shorten >=3pt] (2.8,2.2) -- (4.2,3.6);
  \node[left] at (2.9,2.1) {$x_{u,v}$};
  \node[above] at (0.3,6) {$1$};
  \node[above] at (0.9,6) {$2$};
  \node[above] at (2.4,6) {$\dots$};
  \node[above] at (4.2,6) {$v$};
  \node[above] at (5.7,6) {$K$};
  \node[left] at (-0.2,5.7) {$1$};
  \node[left] at (-0.2,4.8) {$u'$};
  \node[left] at (-0.2,3.6) {$u$};
  \node[left] at (-0.2,2.1) {$\vdots$};
  \node[left] at (-0.2,0.3) {$K$};
  \node[right, text=buy!85!black] at (6.4,4.8)
    {$\circ\circ\circ$\enspace share units of row $u'$};
  \node[right, text=buy!85!black] at (6.4,3.6)
    {$\Out(u)$ -- a row: this you can BUY};
  \node[below, text=ask!85!black] at (4.2,-0.1)
    {$\In(v)$ -- a column: this you are ASKED for};
\end{tikzpicture}

\caption[The labels as a matrix]{The labels as the entries of the
adjacency matrix of $G$. A row $\Out(u)$ can be bought; a column $\In(v)$ is
what a query asks for.}
\label{fig:matrix}
\end{figure}
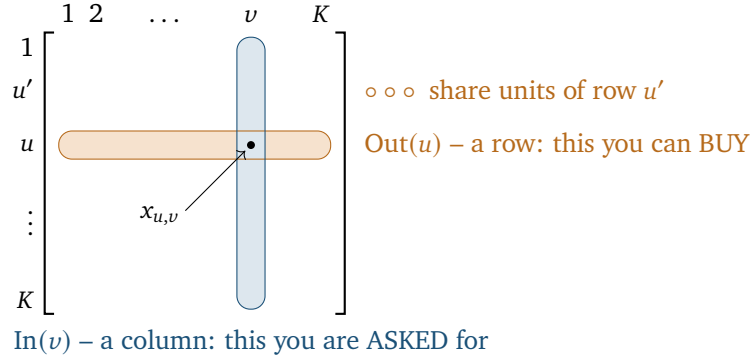

\begin{definition}[The memory game]
\label{def:game}
The game is played in two phases.
\begin{enumerate}[label=\arabic*), itemsep=3pt]
  \item \textbf{Preparation.} Knowing $G$ but not the query, the player
        stores some objects, each taking one unit of memory. A \emph{pure}
        player may store
        \begin{itemize}[itemsep=1pt]
          \item labels $x_{u,v}$;
          \item \emph{share units} for a vertex $u$: if he holds $l_u$ share
                units for $u$ ($0\le l_u\le K$), buying $\Out(u)$ costs him
                $K - l_u$ instead of $K$.
        \end{itemize}
        Let $I$ be the number of stored units.
  \item \textbf{Query.} A vertex $v$ is drawn uniformly from $[K]$ and
        revealed. The player must produce all of $\In(v)$. Besides reading
        what he stored (free), a pure player can obtain a label only by
        \emph{buying a row}: for price $K - l_u$ he obtains all of $\Out(u)$.
        Let $O(v)$ be the total price paid.
\end{enumerate}
The \emph{total memory cost} of a strategy is
\[
  \TMC \;=\; I + \E_v\bigl[O(v)\bigr].
\]
An \emph{entangled} player may, in addition, store XORs
$x_{u_1,v_1}\oplus\dots\oplus x_{u_q,v_q}$ of any labels, one unit each,
however many labels are inside. During the query he may XOR known labels
out of a stored XOR: if all labels of an XOR but one are known, the last
one is recovered for free.
\end{definition}

A share unit is a part of the price of a row, paid in advance: in
\cref{sec:dag} it will be one of the $K$ sources from which the row is
computed.

\begin{intuition}[The game in one sentence.]
You are asked for everything that points INTO $v$, but you can only buy
things in bundles of everything that points OUT of some vertex. All the fun
comes from this mismatch.
\end{intuition}

\begin{intuition}[Why the average?]
$I$ is paid once, before the game, and then the query comes at random. If we
looked at one fixed $v$, a player could prepare exactly for that $v$ (store
all of $\In(v)$ and nothing else) and look much better than he really is.
The average over $v$ is the honest price.
\end{intuition}

\subsection{The pure player pays \texorpdfstring{$2|E|$}{2|E|}}

\begin{theorem}[Pure lower bound in the game]
\label{thm:game-pure}
Every pure strategy has $\TMC \ge 2|E| = Kd$, and this is attained (for
example by storing nothing, or by storing all labels).
\end{theorem}

\begin{proof}
Fix a pure strategy. For every vertex $u$ let $s_u$ be the number of stored
labels from row $u$ and $l_u$ the number of stored share units for $u$; so
$I = \sum_u (s_u + l_u)$.

Consider a query $v$. The player needs every $x_{u,v}$, $u\in N(v)$. If
$x_{u,v}$ is stored, nothing has to be done. If it is not, the only way to
get it is to buy row $u$, at price $K - l_u$. Rows are bought at most once,
and buying any other row does not help with $x_{u,v}$. So the optimal
answer to $v$ buys exactly the rows $u\in N(v)$ whose label $x_{u,v}$ is
not stored, and
\[
  O(v) \;=\; \sum_{\substack{u\in N(v)\\ x_{u,v}\text{ not stored}}} (K - l_u) .
\]
We may assume that the player answers every query optimally, since this
only lowers $\TMC$. Now count from the side of the rows: row $u$ is bought
exactly for the
queries $v\in N(u)$ whose label $x_{u,v}$ is not stored, i.e.\ for
$d_u - s_u$ of the $K$ possible queries. Hence
\begin{align*}
  \TMC
  &= \sum_{u} (s_u + l_u)
     \;+\; \frac1K\sum_{u} (d_u - s_u)(K - l_u) \\
  &= \sum_u d_u \;+\; \sum_u l_u\Bigl(1 - \frac{d_u - s_u}{K}\Bigr)
  \;\ge\; \sum_u d_u \;=\; 2|E| ,
\end{align*}
because $d_u - s_u \le d_u \le K-1 < K$. Storing nothing ($s_u=l_u=0$)
gives equality, and so does storing all labels ($s_u = d_u$, $l_u=0$).
\end{proof}

\begin{intuition}[What the formula says.]
A stored label at best exactly pays for itself: it costs $1$, and it saves
$K$, but only when its column is asked, which happens with probability
$1/K$. A share unit never pays for itself: it costs $1$ and saves $1$, but
only when its row is needed, which is less than always. That is the whole
reason the pure player is stuck: a query asks for a COLUMN, but he can only
buy ROWS, and from every row he buys he uses just ONE entry.
\end{intuition}

\subsection{The entangled player pays for the non-neighbours}

\paragraph{Warm-up: why XORs help at all.} Recall that $x_{u,v}$ sits at
$(u,v)$ and $x_{v,u}$ at the mirror position $(v,u)$ --- so the column
$\In(v)$ we are asked for is the mirror image of the row $\Out(v)$ we can
buy. Let's glue the mirrors together: for every edge $\{u,v\}$ store the
single unit $x_{u,v}\oplus x_{v,u}$. That's $|E|$ units, half of what
storing all labels costs. When asked for $\In(v)$, buy only the own row
$\Out(v)$, at price $K$; then for every neighbour $u$,
\[
  x_{u,v} = (x_{u,v}\oplus x_{v,u})\oplus x_{v,u} .
\]
So $\TMC = |E| + K = Kd/2 + K$, about half of the pure player's $Kd$ when
$d$ is large. To win
\emph{big}, one XOR has to carry many edges at once.

\paragraph{The matching strategy.} A matching $M\subseteq E$ is
\emph{induced} if $G$ has no edges between the vertices of $M$ other than
those of $M$ itself: look only at the vertices touched by $M$, and all you
see is $M$. Suppose
\[
  E = M_1\sqcup\dots\sqcup M_m
\]
is a partition of the edges into $m$ induced matchings. For a vertex $v$
let $H_v = [K]\setminus N(v)$; it contains the non-neighbours of $v$ and $v$
itself.

\begin{theorem}[Entangled upper bound in the game]
\label{thm:game-xor}
The entangled player who stores
\[
  z_j \;=\; \bigoplus_{\{a,b\}\in M_j} \bigl(x_{a,b}\oplus x_{b,a}\bigr),
  \qquad j=1,\dots,m,
\]
and, on query $v$, buys the rows $\Out(h)$ for all $h\in H_v$, answers every
query correctly and has
\[
  \TMC \;=\; m + K(K-d) .
\]
\end{theorem}

\begin{proof}
$I = m$, and on query $v$ he pays $K|H_v| = K(K - d_v)$, which averages to
$K(K-d)$. It remains to check that he can recover every $x_{u,v}$,
$u\in N(v)$. Let $M_j$ be the matching containing $\{u,v\}$
(\cref{fig:matching}). Then
\begin{equation}
\label{eq:game-decode}
  x_{u,v} \;=\; z_j \oplus x_{v,u} \oplus
  \bigoplus_{\substack{\{a,b\}\in M_j\\ \{a,b\}\ne\{u,v\}}} \bigl(x_{a,b}\oplus x_{b,a}\bigr)
\end{equation}
(the edge $\{u,v\}$ is taken out of the big XOR; otherwise its labels would
appear twice and cancel). Every label on the right is known:
\begin{itemize}[itemsep=2pt]
  \item $x_{v,u}\in\Out(v)$, and $v\in H_v$ because $G$ has no loops;
  \item for any other edge $\{a,b\}\in M_j$: it does not touch $u$ or $v$
        ($M_j$ is a matching), and since $M_j$ is \emph{induced} there is
        no edge between $v$ and $a$, nor between $v$ and $b$. So
        $a,b\in H_v$, and we bought $\Out(a)\ni x_{a,b}$ and
        $\Out(b)\ni x_{b,a}$.
\end{itemize}
Two edges at $v$ never lie in the same matching, so every $x_{u,v}$ has its
own XOR.
\end{proof}

\begin{figure}[htbp]
\centering
\begin{tikzpicture}[scale=0.8,
    vert/.style={circle, draw, inner sep=0pt, minimum size=5pt, fill=white},
    hv/.style={vert, fill=buy!45},
    every label/.style={font=\small}]
  \draw[dashed, gray] (0,0) circle (2.35);
  \foreach \a/\n in {90/n1, 140/n2, 250/n4, 295/n5, 340/n6, 30/n7}
    \node[vert] (\n) at (\a:1.6) {};
  \node[vert, label=left:$u$] (u) at (195:1.6) {};
  \node[circle, draw, fill=buy!45, inner sep=1.2pt, font=\small] (v) at (0,0) {$v$};
  \foreach \n in {n1, n2, n4, n5, n6, n7} \draw[gray] (v) -- (\n);
  \node[hv, label=above:$a$] (b) at (-1.1,3.25) {};
  \node[hv, label=above:$b$] (c) at (1.5,3.2) {};
  \node[hv] (a) at (-3.3,1.7) {};
  \node[hv] (d) at (3.45,1.3) {};
  \node[hv] (e) at (3.3,-1.7) {};
  \node[hv] (f) at (1.1,-3.35) {};
  \node[hv, label=below:$c$] (g) at (-1.6,-3.15) {};
  \node[hv, label=left:$d$] (h) at (-3.5,-0.9) {};
  \draw[gray] (a) -- (b);
  \draw[gray] (c) -- (d);
  \draw[gray] (d) -- (e);
  \draw[gray] (a) -- (n2);
  \draw[gray] (e) -- (n6);
  \draw[gray] (f) -- (n5);
  \draw[ask, very thick] (v) -- (u);
  \draw[ask, very thick] (b) -- (c);
  \draw[ask, very thick] (g) -- (h);
  \draw[gray] (-45:2.35) -- (-45:3.0);
  \node[below right, font=\small] at (-45:3.0) {$N(v)\cup\{v\}$};
\end{tikzpicture}

\caption[Decoding with an induced matching]{Decoding $x_{u,v}$
with the XOR $z_j$. \textcolor{ask}{Blue}: the matching $M_j\ni\{u,v\}$; here
$M_j=\bigl\{\{u,v\},\{a,b\},\{c,d\}\bigr\}$, so
$x_{u,v}=z_j\oplus x_{v,u}\oplus x_{a,b}\oplus x_{b,a}\oplus x_{c,d}\oplus x_{d,c}$.
Shaded: $H_v$ --- the vertices whose rows the entangled player buys,
including $v$ itself.}
\label{fig:matching}
\end{figure}
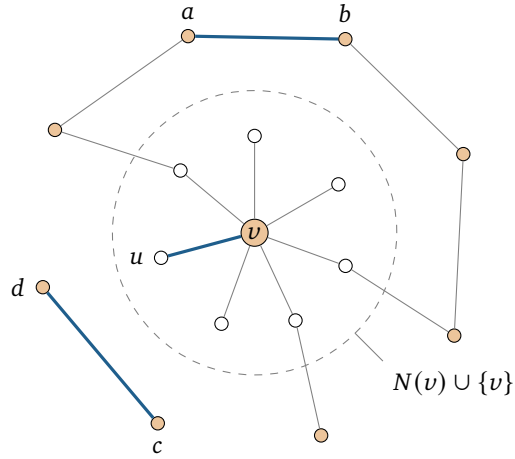

\begin{intuition}[Why induced?]
This is EXACTLY the place where we need the matchings to be induced. With
an ordinary matching, some other edge of $M_j$ could touch the
neighbourhood of $v$; we did not buy that row, so that pair stays locked
inside $z_j$ and we're stuck with at least two unknowns.
\end{intuition}

\begin{intuition}[The whole trick.]
To assemble the column $\In(v)$ the pure player pays for the rows of the
NEIGHBOURS of $v$ ($d$ of them on average). The entangled player pays for
the rows of the NON-neighbours ($K - d$ of them) plus $m$ for the XORs. In a
dense graph non-neighbours are few.
\end{intuition}

\subsection{The graphs we need, and the separation}

For the matching strategy to win by more than a constant we need both
\[
  K - d \ \text{small (almost complete graph)}
  \qquad\text{and}\qquad
  m \ \text{small (few, but large, induced matchings)} .
\]
Here ``small'' is measured against the pure player's $Kd$, which is about
$K^2$ for such dense graphs: we want both $K(K-d)$ and $m$ to be at most
$K^{2-\epsilon}$ for a constant $\epsilon>0$. These two pull in opposite
directions: the denser the graph, the harder it
is to have large induced matchings, since every extra edge may break
inducedness. A complete graph is useless --- its induced matchings have one
edge each, so $m = \binom K2$. Write
$M_{\mathrm{miss}} = \binom K2 - |E|$ for the number of missing edges, so
that
\begin{equation}
\label{eq:missing}
  K(K-d) = K^2 - 2|E| = K + 2M_{\mathrm{miss}} .
\end{equation}
Graphs that balance the two requirements were constructed by Alon, Moitra
and Sudakov~\cite{AMS2013}. In the form we need (\cref{app:ams} derives
the parameters):

\begin{proposition}[AMS graphs]
\label{prop:ams}
Let $\eta = 0.029$. For every large enough $\nu$ there is a regular graph on
$K = 64^\nu$ vertices whose edges are partitioned into $m$ induced matchings
of equal size $\mu=|E|/m$, with
\[
  m \le K^{2-2\eta}, \qquad M_{\mathrm{miss}} \le K^{2-2\eta} .
\]
The graph and the partition can be computed deterministically in time
polynomial in $K$.
\end{proposition}

For these graphs, by \eqref{eq:missing},
\begin{equation}
\label{eq:ams-consequences}
  K(K-d) = O(K^{2-2\eta}),
  \qquad
  d = K - 1 - \frac{2M_{\mathrm{miss}}}{K} = \Theta(K).
\end{equation}

\begin{corollary}[Separation in the game]
\label{cor:game-sep}
On the graphs of \cref{prop:ams}, every pure player has
$\TMC \ge Kd = \Theta(K^2)$, while the entangled player of
\cref{thm:game-xor} has $\TMC = O(K^{2-2\eta})$.
\end{corollary}

\begin{center}
\small
\begin{tabular}{@{}lccc@{}}
  \toprule
  & stores ($I$) & pays on the spot (avg.\ $O$) & TMC \\
  \midrule
  pure player & \multicolumn{2}{c}{\itshape whatever he does\dots} & $\geq 2|E| = Kd$ \\
  entangled: XOR per edge (warm-up) & $|E|$ & $K$ & $Kd/2 + K$ \\
  entangled: XOR per induced matching & $m$ & $K(K-d)$ & $m + K(K-d)$ \\
  \midrule
  on AMS graphs: pure & & & $\Theta(K^2)$ \\
  on AMS graphs: entangled & & & $O(K^{2-2\eta})$ \\
  \bottomrule
\end{tabular}
\end{center}

\paragraph{What is made up, and what comes next.} Two things in the game
are postulated rather than derived:
\begin{enumerate}[label=(\roman*), itemsep=2pt]
  \item \emph{What is a label, and why can't I compute just one?} We simply
        declared that $\Out(u)$ comes as a bundle for price $K$.
  \item \emph{What does it cost to hold memory?} We charged storage once
        ($I$) and the purchases once ($O$), as if memory were free while
        we compute.
\end{enumerate}
In the rest of the paper the labels become hash values computed along a DAG,
which answers (i), and the cost becomes the cumulative cost of pebbling,
which answers (ii). \Cref{sec:models} fixes the definitions and collects
the tools; \cref{sec:dag} builds our DAG.

% ======================================================================
\section{Models and preliminaries}
\label{sec:models}

This section fixes the model and collects the tools used in the rest of
the paper: how labels are computed (\cref{sec:labels}), what pebbles are
and what we pay for (\cref{sec:pebbles,sec:cost}), what an entangled player
may store (\cref{sec:entangled-model}), one counting lemma
(\cref{sec:forced}), and the connectors that will let us reduce the
indegree (\cref{sec:connectors}).

Throughout, $\mathcal D$ is a finite DAG that describes a computation: a
vertex stands for an intermediate value, and an edge
$p\to z$ says that the value at $z$ is computed from the value at $p$. We
write $\Par(z)$ for the set of parents of $z$ (the vertices with an edge
into $z$).

\subsection{Labels are computed, not given}
\label{sec:labels}

Every vertex $z$ of $\mathcal D$ has a value $X_z$ (its \emph{label}),
computed by a hash function from the labels of its parents. To compute a
label you need all of its parents' labels at hand --- there is no shortcut.
Recovering a label that was computed earlier from encoded storage is a
different operation, defined in \cref{sec:entangled-model}. The question is
how much memory you must keep, and for how long, to obtain some target
label.

Concretely, fix a hash function $\mathsf H:\{0,1\}^*\to\{0,1\}^w$, which we
model as a random oracle, and an input $x$. Every vertex $z$ has a public
identifier $\mathsf{id}(z)$, and if $z$ has parents $p_1,\dots,p_r$ (in a
fixed order) then
\begin{equation}
\label{eq:label}
  X_z = \mathsf H\bigl(x,\ \mathsf{id}(z),\ X_{p_1},\dots,X_{p_r}\bigr).
\end{equation}
A vertex without parents is a \emph{source}; its label is
$\mathsf H(x,\mathsf{id}(z))$ and can be computed at any time.
Tuples are encoded unambiguously, and distinct vertices have
distinct identifiers.

All labels are $w$-bit strings (\cref{app:dmhf} assumes
$w\ge2\log_2K$), and $w$ bits is our unit of memory. So a
stored label costs $1$, and so does a stored XOR of any number of labels ---
it is also a $w$-bit string. This is where ``one XOR takes one unit'' from
the game comes from, and why labels behave like the independent random
strings of \cref{def:game}.

\begin{intuition}[Why a random oracle?]
Because then labels behave like the independent random strings of the game.
Knowing some labels tells you nothing about a label you have not computed,
and the only way to compute $X_z$ for the first time is to query $\mathsf H$
on exactly the right input --- which contains all the parents' labels. That is the ``no shortcut''
rule. The identifier $\mathsf{id}(z)$ is there only so that two vertices
with the same parents still get different labels.
\end{intuition}

\begin{remark}[No linearity]
\label{rem:no-linearity}
When we recover a label from an XOR, we XOR $w$-bit strings that
$\mathsf H$ has already returned. We never use anything like
$\mathsf H(A\oplus B) = \mathsf H(A)\oplus\mathsf H(B)$. Every new label is
obtained by an honest query to $\mathsf H$.
\end{remark}

\subsection{Pebbles}
\label{sec:pebbles}

Pebbling is the standard bookkeeping device for such computations. Instead
of tracking bit strings, we put a \emph{pebble} on a vertex to say ``I
currently store its label''. The computation proceeds in rounds; in each
round we can compute many labels in parallel, as long as their parents were
available at the start of the round.

\begin{definition}[Pure pebbling]
\label{def:pure}
A \emph{pure configuration} is a set $P\subseteq V(\mathcal D)$ of pebbled
vertices; its weight is $|P|$. A move $P\to P'$ is \emph{legal} if every
newly pebbled vertex had all its parents pebbled before the move:
\begin{equation}
\label{eq:pure-move}
  z\in P'\setminus P \quad\Longrightarrow\quad \Par(z)\subseteq P .
\end{equation}
Any pebbles may be removed in a move, including the parents just used.
\end{definition}

Sources have no parents, so they can be pebbled at any time. A label first
computed in a move cannot be used as a parent in the same move. The word
``pure'' refers to the fact that memory consists of whole labels: a label is
either stored or not, nothing in between.

\paragraph{A small example.} Take five vertices with
$a\to u\to z$ and $b\to v\to z$, and start with nothing stored. One legal
way to get the label of $z$ is
\[
  \varnothing \;\to\; \{a,b\} \;\to\; \{u,v\} \;\to\; \{z\} .
\]
The move $\varnothing\to\{a,b,u,v\}$ would be illegal: the parents of $u,v$
are not there when the move starts.

\subsection{The cost: one challenge from a prepared state}
\label{sec:cost}

We do not ask the player to compute one fixed label. We ask him to be
\emph{ready}: he prepares some memory state, and then a random challenge
arrives. Fix a nonempty set $\mathcal Q\subseteq V(\mathcal D)$ of possible
challenges.

\begin{definition}[One challenge from a prepared state]
\label{def:service}
\leavevmode
\begin{enumerate}[label=\arabic*), itemsep=2pt]
  \item Before anything happens, the player chooses a configuration $S$ (the
        \emph{prepared state}), reachable from the empty configuration by
        legal moves. Preparing it is \emph{not} charged.
  \item A challenge $q$ is drawn uniformly from $\mathcal Q$, independently
        of $S$, and revealed.
  \item The player makes legal moves $P_0 = S, P_1, P_2,\dots$ The challenge
        is answered at the first time $\tau\ge1$ with $q\in P_\tau$.
  \item The cost is the sum of the weights of the configurations the player
        moved \emph{from}:
        \begin{equation}
        \label{eq:cans}
          \Cans = \sum_{j=0}^{\tau-1} |P_j| .
        \end{equation}
\end{enumerate}
The \emph{pure (single-challenge) cumulative cost} of $(\mathcal D,\mathcal Q)$
is the best expected cost over all prepared states and strategies:
\[
  \CPC_{\mathrm{pure}}(\mathcal D,\mathcal Q)
  \;=\; \inf_{S,\ \text{strategy}} \E[\Cans] .
\]
\end{definition}

This is the \emph{cumulative} memory cost: every stored label is paid for
in every round it is kept --- the ``area under the memory curve''. Note
$\tau\ge1$: even if the answer is already in $S$, you have to make one move,
and you pay $|S|$ for it. So storing things is never free. In the example
above, with $S=\varnothing$ and challenge $z$, the cost is $0+2+2=4$.

\begin{intuition}[Why the cumulative cost and not $I+O$?]
In the memory game we paid for the stored data once. Real memory is paid
for as long as you hold it --- it occupies hardware for the whole time of
the computation. The cumulative cost is the standard measure for
memory-hard functions, exactly for this reason. Nothing breaks: our
strategies take few rounds (a constant number in \cref{sec:dag}, $O(\log K)$
in \cref{sec:degree-two}), so ``paid once'' and ``paid every round'' differ
by at most that factor in the stored part. It does matter in one place: the
warm-up strategy of \cref{sec:game} stops paying off on the DAG
(\cref{sec:entangled}).
\end{intuition}

\paragraph{The trivial upper bound.} The pure player can always prepare
$S=\mathcal Q$ (all possible answers, computed beforehand) and answer in one
move. Hence, for every DAG,
\begin{equation}
\label{eq:trivial}
  \CPC_{\mathrm{pure}}(\mathcal D,\mathcal Q) \;\le\; |\mathcal Q| .
\end{equation}
A lower bound of order $|\mathcal Q|$ therefore says that the pure player
cannot do essentially better than caching all possible answers.

\subsection{Entangled pebbles}
\label{sec:entangled-model}

The entangled player may also store XORs of labels. Alwen
et al.~\cite{AlwenEtAl2016} allow more general ``threshold'' pebbles;
reading an XOR of $|Y|$ labels as their pebble $\langle Y\rangle_{|Y|-1}$,
\cref{rem:ackkpt} checks that every move we make is also a move in their
model.

\begin{definition}[XOR pebbles and recovery]
\label{def:entangled}
An \emph{entangled configuration} $\mathcal P = (W,\mathcal X)$ consists of
a set $W\subseteq V(\mathcal D)$ of pure pebbles and a set $\mathcal X$ of
\emph{XOR pebbles}; an XOR pebble is a set $Y\subseteq V(\mathcal D)$ with
$|Y|\ge2$ and stands for the stored string $\bigoplus_{z\in Y}X_z$. The
weight is
\[
  \wt(\mathcal P) = |W| + |\mathcal X| .
\]
The \emph{closure} $\cl(\mathcal P)$ is the set of labels we can recover
for free: the smallest set $C\supseteq W$ such that
\[
  Y\in\mathcal X \ \text{ and } \ |Y\setminus C|\le 1
  \quad\Longrightarrow\quad Y\subseteq C .
\]
(If all labels of an XOR but one are known, XOR them out to get the last
one.) The labels \emph{available} in a move are those recoverable or
computable from recoverable parents:
\[
  \avail(\mathcal P) = \cl(\mathcal P)\ \cup\ \{z : \Par(z)\subseteq\cl(\mathcal P)\}.
\]
A move $\mathcal P\to\mathcal P'$ is legal if every new pure pebble is in
$\avail(\mathcal P)$ and every new XOR pebble $Y$ satisfies
$Y\subseteq\avail(\mathcal P)$. Pebbles may be kept or removed
freely. As in pure pebbling, availability refers to $\mathcal P$, the
configuration before the move: pebbles removed in a move may still be used
in it, and a label first computed in a move cannot be used in that move as
a parent. It may help recover other labels in the resulting configuration,
but such a recovery cannot justify further placements within the same move.
\end{definition}

For example, in the five-vertex DAG of \cref{sec:pebbles}, take
$\mathcal P=(\{u\},\{\{u,v\}\})$: a pure pebble on $u$ and the XOR
$X_u\oplus X_v$. Then $\cl(\mathcal P)=\{u,v\}$, so $z$ is available and
one move can pebble it, although $v$ is not stored as a whole. Without the
pure pebble on $u$ the closure would be empty: the XOR alone gives nothing.
Starting instead from $(\{a\},\{\{u,v\}\})$, one move can pebble $u$, but
that move cannot also pebble $z$: $v$ becomes recoverable only in the
resulting configuration.

A pure configuration is the case $\mathcal X=\varnothing$; then
$\cl(\mathcal P) = W$ and the rule is exactly \eqref{eq:pure-move}. For
the entangled player, the challenge is answered at the first $\tau\ge1$
with $q\in\cl(\mathcal P_\tau)$, and $\Cans = \sum_{j<\tau}\wt(\mathcal P_j)$.

\begin{intuition}[Stored vs.\ recoverable.]
Recovering a label from an XOR is free (it's a local XOR of strings we
have), but it does not \emph{store} anything: the recovered label costs
nothing and is not kept unless we put a pure pebble on it, which costs $1$
per round like any other pebble. Recovery also never computes a new label;
for that you still need a move with all the parents available.
\end{intuition}

\begin{remark}[Our moves are moves of Alwen et al.]
\label{rem:ackkpt}
Alwen et al.~\cite{AlwenEtAl2016} store \emph{threshold pebbles}
$\langle Y\rangle_\theta$, with $Y\subseteq V(\mathcal D)$ and
$0\le\theta\le|Y|$, of weight $|Y|-\theta$: any $\theta$ labels of $Y$
give back all of $Y$. Their closure $\cl^\ast(\mathcal P)$ is the smallest
set $C$ such that $|Y\cap C|\ge\theta$ implies $Y\subseteq C$ for every
stored $\langle Y\rangle_\theta$, and $\mathcal P'$ is a valid update of
$\mathcal P$~\cite[Definitions~1--2]{AlwenEtAl2016} if every pebble
$\langle Y'\rangle_{\theta'}$ of $\mathcal P'$
\begin{enumerate}[label=(\roman*), itemsep=1pt]
  \item satisfies $Y'\subseteq\cl^\ast(\mathcal P)\cup\{z:\Par(z)\subseteq\cl^\ast(\mathcal P)\}$, or
  \item is a pebble $\langle Y'\rangle_{\theta}$ of $\mathcal P$ whose
        threshold is kept or raised, $\theta'\ge\theta$.
\end{enumerate}
Read a pure pebble on $z$ as $\langle\{z\}\rangle_0$ and an XOR pebble $Y$
as $\langle Y\rangle_{|Y|-1}$ (an XOR is one way to realize such a
pebble). Both weigh one, and the threshold rules become ``$z\in C$'' and
``$|Y\setminus C|\le1\Rightarrow Y\subseteq C$''; so $\cl^\ast=\cl$ on our
configurations, and the set in (i) is $\avail(\mathcal P)$. Hence in a
legal move of \cref{def:entangled} every new pebble satisfies (i) and every
kept pebble satisfies (ii) with $\theta'=\theta$; conversely, between
configurations of pure and XOR pebbles their rules allow nothing more.
Apart from general threshold pebbles, the one freedom we never use is
raising a threshold, which would turn a pure or XOR pebble into a
weight-zero pebble that recovers nothing. This also covers the moves below
that may look unusual: computing a label whose parents are only
recoverable (move $2\to3$ in \cref{lem:xor}), storing a recovered label
whose own parent is not available (Decode in \cref{sec:entangled-two}),
and dropping helpers in the move that uses them. So every strategy in this
paper is a strategy in their model, with the same cost and the same answer
times: all our entangled upper bounds hold there, and pure configurations
are ordinary pebbling, \eqref{eq:pure-move}. Nothing is claimed in the
other direction; we prove no entangled lower bound.
\end{remark}

In a proof of space or in a memory-hard function, the prepared state is
used again and again. So it is natural to ask the player to answer the
challenge \emph{and then restore his prepared state}, ready for the next
challenge.

\begin{definition}[Returning service]
\label{def:return}
A strategy \emph{returns} if, after answering at time $\tau$, it continues
until some time $T\ge\tau$ with $\mathcal P_T = S$. Its return cost is
\[
  \Cret = \sum_{j=0}^{T-1}\wt(\mathcal P_j) \ \ \ge\ \Cans ,
  \qquad
  \CPC^{\mathrm{ret}}_{\mathrm{ent}}(\mathcal D,\mathcal Q)
  = \inf_{S,\ \text{returning strategy}} \E[\Cret] .
\]
\end{definition}

The configuration $\mathcal P_\tau$ that holds the answer is charged in
$\Cret$ when $T>\tau$, but not in $\Cans$; the final configuration
$\mathcal P_T=S$ is never charged.

We will give the pure player the easier task (answer only) and the
entangled player the harder one (answer and return). The pure lower bounds
hold for every prepared state; the entangled upper bounds are for one
specific state, and for a very restricted kind of strategy.

\begin{definition}[Common arity]
\label{def:common-arity}
A returning strategy is \emph{common-arity} if its prepared state can be
reached from the empty configuration by legal moves such that all XOR
pebbles stored along the way, and all XOR pebbles stored while answering
and returning any challenge, have the same number of elements, its
\emph{arity}. In the repeated game of \cref{sec:repeated} the same
condition is imposed on all phases. The arity may depend on the instance
but not on the challenge; pure pebbles are not restricted.
\end{definition}

Pure strategies are the trivial case. Common arity is a property of the
strategies behind our upper bounds (\cref{sec:entangled}); we prove no
lower bound for such strategies.

\subsection{Forced ancestors}
\label{sec:forced}

All our pure lower bounds rest on one simple observation: if a vertex
leads to the challenge along a path that the cache does not cut, then the
player has to compute it again.

\begin{lemma}[Forced ancestors]
\label{lem:ancestors}
Let $P_0,P_1,\dots,P_\tau$ be legal pure moves, and let $c\notin P_0$
first appear in $P_\tau$. Let $U$ be a set of ancestors of $c$, each of
which has a directed path to $c$ avoiding $P_0$. Then every $u\in U$
appears in some $P_j$ with $1\le j<\tau$, and so
\[
  \sum_{j=1}^{\tau-1}|P_j| \ \ge\ |U| .
\]
\end{lemma}

\begin{proof}
Fix $u\in U$ and a path $u=z_0\to z_1\to\dots\to z_J = c$ avoiding $P_0$.
Walk backwards. $c$ is first placed at time $a_J=\tau$; by
\eqref{eq:pure-move} its parent $z_{J-1}$ is present at time $\tau-1$, so
$z_{J-1}$ first appears at some time $a_{J-1}\le\tau-1$, and
$a_{J-1}\ge1$ because $z_{J-1}\notin P_0$. Repeating, we get
\[
  1\le a_0<a_1<\dots<a_J=\tau ,
\]
so $u\in P_{a_0}$ with $1\le a_0<\tau$. Different $u$ give different pairs
$(u,a_0)$ (vertex, round), each contributing $1$ to the sum.
\end{proof}

\subsection{Connectors}
\label{sec:connectors}

The last ingredient is a family of graphs that we will use in
\cref{sec:degree-two} to replace a vertex with $K$ parents by a network of
indegree two. What we need from such a network is that a few stored
vertices inside it cannot cut many of its input--output connections. The
right notion is due to Blocki and Cinkoske~\cite{BlockiCinkoske2021}.

\begin{definition}[$K$-connector]
\label{def:connector}
A \emph{$K$-connector} is a DAG $M$ with $K$ designated input terminals
$I$ and $K$ output terminals $O$ such that for every $r\le K$, every
$I'\subseteq I$ and $O'\subseteq O$ of size $r$, and every bijection
$\pi:I'\to O'$, there are $r$ pairwise vertex-disjoint directed paths in
$M$, one from each $i\in I'$ to $\pi(i)$.
\end{definition}

So a connector can route any pairing of inputs and outputs at the same time,
without two routes sharing a vertex. For the lower bound we use a
consequence of this: after deleting any $q\le K$ vertices, we can still
select $K-q$ inputs and $K-q$ outputs such that every selected input
reaches every selected output. The selection may leave out some terminals
that were not deleted.

\begin{definition}[ST-robustness~\cite{BlockiCinkoske2021}]
\label{def:st}
A DAG $M$ with input set $I$ and output set $O$, $|I|=|O|=K$, is
\emph{maximally ST-robust} if for every $q\le K$ and every set
$D\subseteq V(M)$ with $|D|\le q$ there are $I^*\subseteq I\setminus D$ and
$O^*\subseteq O\setminus D$, each of size at least $K-q$, such that every
$i\in I^*$ has a directed path to every $o\in O^*$ in $M-D$.
\end{definition}

\begin{lemma}[{\cite[Theorem~11]{BlockiCinkoske2021}}]
\label{lem:connector-st}
Every $K$-connector is maximally ST-robust.
\end{lemma}

\begin{intuition}[What ST-robustness buys us.]
Think of $D$ as the pebbles a player keeps inside $M$ (the
\emph{checkpoints}). ST-robustness says: $q$ checkpoints together spoil at
most $q$ inputs and $q$ outputs (the terminals left out of the selection,
which need not be checkpoints themselves); between all the others, there
are still paths that avoid every checkpoint. So a small cache cannot
``summarize'' the network. (Deleting $D$ happens only in the analysis: the
player deletes nothing, and the paths we find simply avoid the vertices he
keeps.)
\end{intuition}

The connector we use is the classical back-to-back butterfly
(Bene\v{s} network~\cite{Benes1964}). Let $K=2^k$. A \emph{butterfly} has
layers $V_0,\dots,V_k$ of $K$ vertices each, indexed by binary addresses
$a\in\{0,1\}^k$; for each stage $j<k$, the vertex $a$ of layer $j$ has
edges to the vertices $a$ and $a\oplus 2^j$ of layer $j+1$ (keep or flip
bit $j$). The \emph{back-to-back butterfly} $\Bk$ glues a butterfly to a
mirrored copy that flips the bits in the reverse order; the two halves
share the middle layer (\cref{fig:butterfly}).

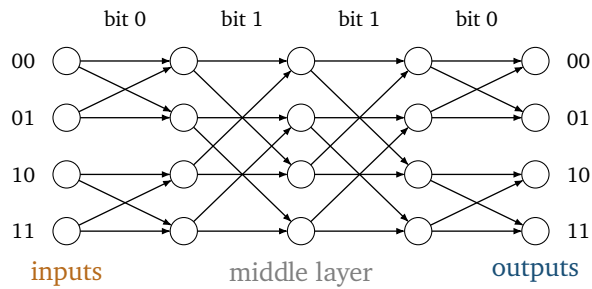
\begin{figure}[htbp]
\centering
\begin{tikzpicture}[x=1.55cm,y=0.75cm,font=\small,
    bv/.style={circle, draw, fill=white, inner sep=0pt, minimum size=3.6mm},
    be/.style={-{Latex[length=1.3mm,width=0.9mm]}, line width=0.5pt, shorten <=1.6mm, shorten >=1.6mm}]
  % stages with bit 0: 0->1 and 3->4 ; bit 1: 1->2 and 2->3
  \foreach \st in {0,3} {
    \foreach \a/\b in {0/1,1/0,2/3,3/2} {
      \draw[be] (\st,-\a) -- (\st+1,-\a);
      \draw[be] (\st,-\a) -- (\st+1,-\b);
    }
  }
  \foreach \st in {1,2} {
    \foreach \a/\b in {0/2,1/3,2/0,3/1} {
      \draw[be] (\st,-\a) -- (\st+1,-\a);
      \draw[be] (\st,-\a) -- (\st+1,-\b);
    }
  }
  \foreach \st in {0,...,4} \foreach \a in {0,...,3} \node[bv] at (\st,-\a) {};
  \foreach \a/\lab in {0/00,1/01,2/10,3/11} {
    \node[anchor=east] at (-0.18,-\a) {\scriptsize $\lab$};
    \node[anchor=west] at (4.18,-\a) {\scriptsize $\lab$};
  }
  \node at (0.5,0.75) {\scriptsize bit $0$};
  \node at (1.5,0.75) {\scriptsize bit $1$};
  \node at (2.5,0.75) {\scriptsize bit $1$};
  \node at (3.5,0.75) {\scriptsize bit $0$};
  \node[below, buy!85!black] at (0,-3.35) {inputs};
  \node[below, gray] at (2,-3.35) {middle layer};
  \node[below, ask!85!black] at (4,-3.35) {outputs};
\end{tikzpicture}

\caption[The back-to-back butterfly]{The back-to-back butterfly
$\mathcal B_4$. Every non-input vertex has two parents.}
\label{fig:butterfly}
\end{figure}

\begin{proposition}[{Butterfly connector~\cite{BlockiCinkoske2021}}]
\label{prop:butterfly}
For every power of two $K=2^k$, the back-to-back butterfly $\Bk$ is a
$K$-connector, hence maximally ST-robust, and
\[
  \indeg(\Bk)\le 2,
  \qquad
  |V(\Bk)| = (2k+1)K = \Theta(K\log K),
  \qquad
  \text{depth } \Delta_K = 2k = O(\log K),
\]
with $K$ vertices in each layer.
\end{proposition}

(A single one-way butterfly is \emph{not} enough: it has essentially one
route between each input and output, and one checkpoint can cut many of
them. The mirrored half is what gives the routing freedom.)

% ======================================================================
\section{Playing the game on a DAG}
\label{sec:dag}

We now make the memory game real. This section builds a DAG on which the
pebbling game of \cref{sec:models} reproduces the memory game: every rule of
\cref{def:game} becomes a property of the graph, the cost becomes the
cumulative cost, and both proofs of
\cref{sec:game} go through with small changes. We do not worry about the
indegree yet; that is \cref{sec:degree-two}.

\subsection{The DAG}
\label{sec:dag-def}

Now we build the DAG that turns the memory game into a computation. As
before, $G=([K],E)$ is an undirected graph with degrees $d_u$, average degree
$d = 2|E|/K$, and a partition of $E$ into $m$ induced matchings. $G$ itself
is \emph{not} the graph we compute on --- it is an index that tells us which
labels exist and which labels share an XOR. The computation happens on the
following DAG $\Hk$, which has three layers.

\begin{definition}[The DAG $\Hk$]
\label{def:dag}
\leavevmode
\begin{enumerate}[label=\textbf{\arabic*.}, leftmargin=*, itemsep=3pt]
  \item \textbf{Sources.} For every $u\in[K]$ a private block of $K$ sources
        \[
          A_u = \{a_{u,1},\dots,a_{u,K}\}.
        \]
  \item \textbf{Edge vertices.} For every \emph{oriented} edge $(u,v)$, i.e.\
        $\{u,v\}\in E$, a vertex $x_{u,v}$ whose parents are \emph{all} $K$
        sources of $A_u$:
        \[
          \Par(x_{u,v}) = A_u .
        \]
  \item \textbf{Outputs.} For every pair $(v,s)\in[K]\times[K]$ an output
        $y_{v,s}$ whose parents are the whole column $v$:
        \[
          \Par(y_{v,s}) = \{x_{u,v} : u\in N(v)\}.
        \]
\end{enumerate}
The challenge set is the set of outputs, $\Qk = \{y_{v,s}:(v,s)\in[K]^2\}$.
\end{definition}

So $\Hk$ has $K^2$ sources, $2|E| = Kd$ edge vertices and $K^2$ outputs:
\[
  |V(\Hk)| = 2K^2 + Kd .
\]
We write $X_{u,v}$ for the label of $x_{u,v}$ and group these labels exactly
as in the game:
\[
  \Out(u) = \{X_{u,v} : v\in N(u)\} \quad\text{(row $u$)},
  \qquad
  \In(v) = \{X_{u,v} : u\in N(v)\} \quad\text{(column $v$)}.
\]

\begin{figure}[htbp]
\centering
\begin{tikzpicture}[x=1cm,y=1cm,font=\small,
    src/.style={draw=buy!85!black, fill=buy!14, rounded corners=2pt,
                minimum width=1.9cm, minimum height=0.5cm, inner sep=2pt},
    edg/.style={circle, draw, fill=white, inner sep=1.2pt, minimum size=0.75cm},
    outp/.style={draw=ask!85!black, fill=ask!12, rounded corners=2pt,
                minimum width=1.0cm, minimum height=0.55cm, inner sep=2pt},
    arr/.style={-{Latex[length=1.6mm,width=1.1mm]}, line width=0.6pt}]
  \node[anchor=east, gray] at (-1.3,  2.6) {sources};
  \node[anchor=east, gray] at (-1.3,  0.8) {edge vertices};
  \node[anchor=east, gray] at (-1.3, -1.2) {outputs};
  \foreach \i/\u/\xx in {1/u_1/0, 2/u_2/3.2, 3/u_3/6.4} {
    \node[src] (A\i) at (\xx,2.6) {$A_{\u}$ \ ($K$ sources)};
    \node[edg] (x\i) at (\xx,0.8) {$x_{\u,v}$};
    \foreach \dx in {-0.75,-0.4,-0.05,0.3,0.65}
      \draw[arr, buy!70!black] (\xx+\dx,2.33) -- (x\i);
  }
  \node at (8.2,2.6) {$\cdots$};
  \foreach \s/\xx in {1/1.4, 2/3.2, K/5.0} {
    \node[outp] (y\s) at (\xx,-1.2) {$y_{v,\s}$};
    \foreach \i in {1,2,3} \draw[arr, ask!80!black] (x\i) -- (y\s);
  }
  \node at (4.1,-1.2) {$\cdots$};
  \draw[dashed, gray, rounded corners=4pt] (-0.75,0.32) rectangle (7.15,1.28);
  \node[anchor=west, gray, font=\scriptsize, align=left] at (7.25,0.8)
    {column $\In(v)$:\\ one vertex per\\ neighbour $u\in N(v)$};
\end{tikzpicture}

\caption[The DAG around a challenge]{The part of $\Hk$ that
matters for a challenge $y_{v,s}$. Each $x_{u,v}$ needs the whole source
block $A_u$; each output of column $v$ needs the whole column.}
\label{fig:hk}
\end{figure}
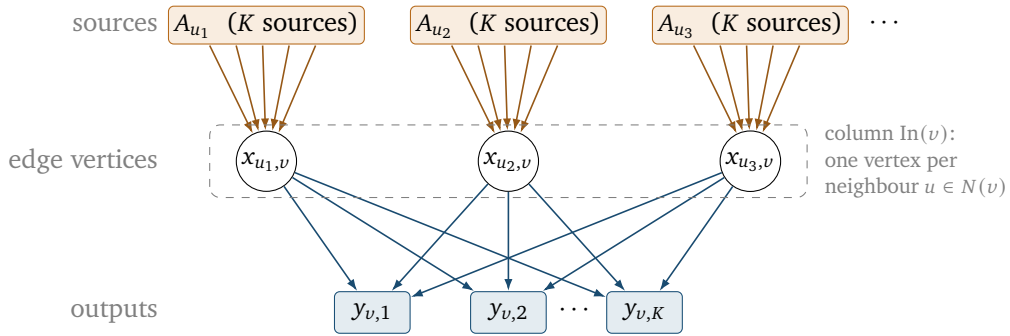

How does this reproduce the rules of the game? Look at what it takes to
compute a label $X_{u,v}$ (\cref{fig:hk}): you need all of $A_u$ at hand.
But once you have $A_u$, you can compute \emph{every} $x_{u,v'}$ in row $u$
at the same time. So a row is bought as a whole, and its price is $K$ ---
the $K$ sources of $A_u$ you have to hold. A stored source of $A_u$ is a
``share unit'' of row $u$: it is one of the $K$ sources you need for the
row, it costs one unit, and it helps with row $u$ only. The whole
dictionary:

\begin{center}
\small
\begin{tabular}{@{}ll@{}}
  \toprule
  memory game & DAG $\Hk$ \\
  \midrule
  row $\Out(u)$ for price $K$ & source block $A_u$, the $K$ parents of every $x_{u,\cdot}$ \\
  share unit for $u$ & a stored source $a_{u,i}\in A_u$ \\
  stored label $x_{u,v}$ & stored label $X_{u,v}$ of the edge vertex $x_{u,v}$ \\
  query $\In(v)$ & challenge $y_{v,s}$, whose parents are the column $\In(v)$ \\
  XOR $z_j$ of a matching & stored XOR of the $2|M_j|$ labels of $M_j$ \\
  \bottomrule
\end{tabular}
\end{center}

\begin{intuition}[Why $K$ outputs per column?]
In the game the answer to a query was the column itself, $d_v$ labels. If
we put just one output $y_v$ on top of every column, the pure player would
simply compute all $K$ outputs in advance, keep them, and answer every
challenge at price $K$ --- much less than $Kd$ (this is \eqref{eq:trivial}
with $|\mathcal Q| = K$). With $K$ different outputs
$y_{v,1},\dots,y_{v,K}$ per column (same parents, different identifiers, so
different labels) there are $K^2$ possible answers, and caching them all
costs $K^2\ge Kd$. The extra index $s$ does nothing else: whatever you do
for $y_{v,s}$, you do for column $v$.
\end{intuition}

% ======================================================================
\subsection{Pure player: still pays \texorpdfstring{$Kd$}{Kd}}
\label{sec:pure}

In the game, the pure lower bound came from a simple fact: from every row
you buy, you use only one label. In the DAG the same fact appears through
\cref{lem:ancestors}: every source that leads to the challenge along a path
avoiding the cache must be recomputed.

\begin{lemma}[Pure lower bound on $\Hk$]
\label{lem:pure}
Let $P$ be any pure prepared state on $\Hk$, $p = |P|$. For a uniform
challenge $y_{v,s}\in\Qk$, every pure strategy satisfies
\[
  \E[\Cans] \ \ge\ p + (Kd - p)_+ \ \ge\ Kd,
  \qquad (z)_+ := \max\{z,0\}.
\]
\end{lemma}

\begin{proof}
Fix the challenge $y_{v,s}$. Its \emph{candidate} sources are the sources
of the blocks $A_u$, $u\in N(v)$; there are $K d_v$ of them. Call a
candidate $a\in A_u$ \emph{certified} if the whole path
\[
  a \longrightarrow x_{u,v} \longrightarrow y_{v,s}
\]
avoids $P$, and let $U(v,s)$ be the set of certified sources. By
\cref{lem:ancestors},
\begin{equation}
\label{eq:pathwise}
  \Cans = p + \sum_{j=1}^{\tau-1}|P_j| \ \ge\ p + |U(v,s)| .
\end{equation}
(If $y_{v,s}\in P$, then $U(v,s)=\varnothing$ and \eqref{eq:pathwise} just
says $\Cans\ge p$, which is true since $\tau\ge1$.)

Now average over the challenge. On average there are
$\E_v[Kd_v] = Kd$ candidates, and a candidate loses its certificate only if
some vertex on its path is cached. Let's see how many certificates one
cached vertex can destroy, on average over the uniform $(v,s)$ (primes mark
the indices of the cached vertex, to keep them apart from the challenge):

\begin{center}
\small
\begin{tabular}{@{}llll@{}}
  \toprule
  cached vertex & candidates it kills & when? & expected loss \\
  \midrule
  source $a\in A_u$ \ (``share unit'') & $1$ (itself) & $v\in N(u)$: prob.\ $d_u/K$ & $\le 1$ \\
  edge vertex $x_{u,v'}$ \ (``label'') & $K$ (all of $A_u$) & column is $v'$: prob.\ $1/K$ & $1$ \\
  output $y_{v',s'}$ & $Kd_{v'}$ (all) & challenge is $y_{v',s'}$: prob.\ $1/K^2$ & $d_{v'}/K\le 1$ \\
  \bottomrule
\end{tabular}
\end{center}

So each of the $p$ cached vertices destroys at most one certificate in
expectation (summing over cached vertices may count a lost certificate
twice, which only helps). Hence $\E[|U(v,s)|]\ge Kd - p$, and since
$|U(v,s)|\ge 0$ also $\E[|U(v,s)|]\ge (Kd-p)_+$. Taking expectations in
\eqref{eq:pathwise} gives the first inequality; the second one is checked
separately for $p\le Kd$ and $p\ge Kd$. Nothing was assumed about how the
player answers, so this holds for randomized strategies too.
\end{proof}

\begin{intuition}[It's the game's proof.]
Look at the first two rows of the table: these are exactly the two facts
behind \cref{thm:game-pure}. A share unit (cached source) saves at most one
unit, and only when its row is needed; a stored label saves $K$, but only
for one query column out of $K$, so it at best pays for itself. The DAG
adds only the third row --- caching an answer --- and the $K$ outputs per
column make sure it doesn't help either.
\end{intuition}

The bound is tight: prepare all $Kd$ edge vertices; every output has all
its parents there, so one move answers, at cost $Kd$. Therefore, for
\emph{every} graph $G$,
\begin{equation}
\label{eq:pure-exact}
  \CPC_{\mathrm{pure}}(\Hk,\Qk) \;=\; Kd \;=\; 2|E| .
\end{equation}
Same answer as in the memory game.

% ======================================================================
\subsection{Entangled player: pays for the non-neighbours}
\label{sec:entangled}

\subsubsection{The prepared state}

For each matching $M_j$ let
\[
  \mathcal S_j = \{x_{a,b},\,x_{b,a} : \{a,b\}\in M_j\}
\]
be the set of its $2|M_j|$ edge vertices, and let the prepared state be one
XOR pebble per matching:
\[
  \Sx = (\varnothing,\ \{\mathcal S_1,\dots,\mathcal S_m\}),
  \qquad
  \wt(\Sx) = m,
  \qquad
  \text{the stored strings are } Z_j = \bigoplus_{z\in\mathcal S_j} X_z .
\]
This is the $z_j$ of the game. $\Sx$ does not depend on the challenge, and
it can be prepared legally: pebble all sources, then all edge vertices, then
place the $m$ XOR pebbles and remove everything else in one move.
Every $\mathcal S_j$ has $2|M_j|\ge2$ elements (blocks of a partition are
nonempty), as \cref{def:entangled} requires. In the strategies of
\cref{lem:xor,lem:xor-two} and the entangled strategy of
\cref{thm:repeated}, including the preparation of $\Sx$, these $m$ XOR
pebbles are created once and kept unchanged, no other
XOR pebble is ever stored, and every other pebble is pure. So when the
matchings have equal size $\mu$, as in \cref{prop:ams}, these strategies
are common-arity of arity $2\mu$ (\cref{def:common-arity}).

\subsubsection{Decoding a column}

For a column $v$, let $H_v = [K]\setminus N(v)$ be the non-neighbours of
$v$ together with $v$ itself, and let
\[
  A(H_v) = \bigcup_{h\in H_v} A_h ,
  \qquad
  D_v = \{x_{h,z} : h\in H_v,\ z\in N(h)\}
\]
be their source blocks and their rows. $|A(H_v)| = K|H_v|$ and
$|D_v| = \sum_{h\in H_v} d_h \le K|H_v|$.

The picture is \cref{fig:matching}: we recompute the rows of
$H_v$, shaded there.

\begin{lemma}[One unknown per XOR]
\label{lem:decode}
For every $v\in[K]$, every label of column $v$ is recoverable from the
prepared XORs together with the rows of $H_v$:
\[
  \{x_{u,v} : u\in N(v)\} \ \subseteq\ \cl\bigl(\Sx\ \cup\ \text{pure pebbles on } D_v\bigr).
\]
\end{lemma}

\begin{proof}
This is the argument of \cref{thm:game-xor}, now phrased with closures.
Take $u\in N(v)$ and the matching $M_j\ni\{u,v\}$. Then
\begin{equation}
\label{eq:decode}
  X_{u,v} = Z_j \oplus X_{v,u} \oplus
  \bigoplus_{\substack{\{a,b\}\in M_j\\ \{a,b\}\ne\{u,v\}}} (X_{a,b}\oplus X_{b,a}) .
\end{equation}
We check that every label on the right except $Z_j$ is in $D_v$.
\begin{itemize}[itemsep=2pt]
  \item $x_{v,u}\in D_v$, because $v\in H_v$ ($G$ has no loops, so
        $v\notin N(v)$).
  \item Take another edge $\{a,b\}\in M_j$. It does not touch $u$ or $v$
        ($M_j$ is a matching), and since $M_j$ is \emph{induced}, $G$ has no
        edge between $v$ and $a$, nor between $v$ and $b$. So
        $a,b\in H_v$, and both $x_{a,b}$ and $x_{b,a}$ are in $D_v$.
\end{itemize}
The target $x_{u,v}$ itself is not in $D_v$, since $u\notin H_v$. So the
XOR pebble $\mathcal S_j$ has exactly one element outside $D_v$, and the
closure rule recovers it. Different neighbours $u$ of $v$ lie in different
matchings (two edges at $v$ cannot be in one matching), so every target is
recovered from its own XOR.
\end{proof}

Note that the lemma talks only about the index structure ($G$, the
matchings, and which edge vertices are pebbled), not about how the edge
vertices are computed. We will use it again, unchanged, in
\cref{sec:degree-two}.

\subsubsection{The schedule}

\begin{lemma}[Entangled upper bound on $\Hk$]
\label{lem:xor}
There is a legal strategy that starts in $\Sx$, answers every challenge
$y_{v,s}$ and returns to exactly $\Sx$ in four moves, with
\[
  \Cret(v,s) \;=\; 4m + K|H_v| + \sum_{h\in H_v} d_h + 1
  \;\le\; 4m + 2K|H_v| + 1 .
\]
On average over the challenge,
\[
  \E[\Cret] \ \le\ 4m + 2K(K-d) + 1 .
\]
\end{lemma}

\begin{proof}
The configurations are (the XOR pebbles of $\Sx$ are kept all the
time, and no other XOR pebble is ever placed):

\begin{center}
\small
\begin{tabular}{@{}cll@{}}
  \toprule
  round & configuration & weight (paid) \\
  \midrule
  $0$ & $\Sx$ & $m$ \\
  $1$ & $\Sx$ + sources $A(H_v)$ & $m + K|H_v|$ \\
  $2$ & $\Sx$ + rows $D_v$ & $m + \sum_{h\in H_v} d_h$ \\
  $3$ & $\Sx$ + the answer $y_{v,s}$ & $m + 1$ \\
  $4$ & $\Sx$ & --- (stop) \\
  \bottomrule
\end{tabular}
\end{center}

Move $0\to1$ pebbles sources, which is always legal. In move $1\to2$ every
vertex of $D_v$ has its parents (a block $A_h$, $h\in H_v$) pebbled, so the
rows can be computed and the sources removed. By \cref{lem:decode}, every
parent of $y_{v,s}$ is in the closure of configuration $2$, so move $2\to3$
can compute the answer and remove the rows; we do not need to materialize
the column. The answer is delivered at $\tau=3$, and move $3\to4$ removes
it. Summing the weights of rounds $0$--$3$ gives the formula; the
inequality uses $d_h\le K$. Finally $\E_v|H_v| = K - d$.
\end{proof}

Compare with \cref{thm:game-xor}, where the entangled player paid $m + K(K-d)$. The
factor $4$ in front of $m$ is the cumulative cost: the XORs are held for
four rounds. The factor $2$ in front of $K(K-d)$ is because buying a row now
happens in two steps --- first we hold the sources, then the row itself.

\begin{intuition}[Where did the warm-up go?]
In the game (the warm-up of \cref{sec:game}), one XOR per edge, $x_{u,v}\oplus x_{v,u}$, already beat the
pure player by a factor of about $2$. In the DAG game it doesn't: the $|E| =
Kd/2$ stored XORs are paid in every round, and answering takes three rounds
(sources $A_v$, row $\Out(v)$, answer), so $\Cans \ge 3\cdot Kd/2 > Kd$.
Constant-factor tricks get eaten by the cumulative cost. What survives is
the matching trick, where the stored state $m$ is \emph{polynomially}
smaller than what the pure player pays.
\end{intuition}

% ======================================================================
\subsection{The separation at high indegree}
\label{sec:separation}

Putting \cref{lem:pure,lem:xor} side by side:
\[
  \text{pure: } \ \E[\Cans]\ \ge\ Kd ,
  \qquad\qquad
  \text{entangled: } \ \E[\Cret]\ \le\ 4m + 2K(K-d) + 1 .
\]
Plugging in the AMS graphs of \cref{prop:ams}, exactly as in
\cref{cor:game-sep}:

\begin{theorem}[Separation on $\Hk$]
\label{thm:sep}
For the graphs of \cref{prop:ams}, the DAG $\Hk$ has
$N = |V(\Hk)| = \Theta(K^2)$ vertices and, for a uniformly random challenge,
\[
  \CPC_{\mathrm{pure}}(\Hk,\Qk) = Kd = \Theta(N),
  \qquad
  \CPC^{\mathrm{ret}}_{\mathrm{ent}}(\Hk,\Qk) = O(K^{2-2\eta}) = O(N^{1-\eta}) .
\]
So the entangled player wins by a polynomial factor $\Omega(N^{\eta})$,
even though he also restores his state and the pure player does not.
The entangled bound is achieved by a common-arity strategy
(\cref{def:common-arity}) of arity $2\mu=Kd/m=\Omega(K^{2\eta})$.
\end{theorem}

\begin{proof}
$N = 2K^2 + Kd = \Theta(K^2)$ by \eqref{eq:ams-consequences}. The pure
cost is \eqref{eq:pure-exact}; the entangled one is \cref{lem:xor} with
\eqref{eq:ams-consequences}.
\end{proof}

Note that the pure bound is as strong as it gets: by \eqref{eq:trivial},
$\CPC_{\mathrm{pure}}\le|\Qk| = K^2$, and we have $Kd=\Theta(K^2)$.

% ======================================================================
\section{Bringing the indegree down to two}
\label{sec:degree-two}

$\Hk$ has vertices with $K$ parents: every $x_{u,v}$ reads all of $A_u$, and
every $y_{v,s}$ reads the whole column. Real constructions use indegree
two, and in the pebbling world a large indegree is a kind of cheating ---
one round of computation on $K$ parents hides a lot of work. In this section
we realize the same dependencies with indegree two, and show that both
bounds survive, up to a logarithmic factor.

\subsection{Why not just binary trees?}
\label{sec:why-not-trees}

The obvious idea is to replace each big fan-in by a binary tree. That does
not work, and seeing why tells us what we need.

Suppose row $u$ is computed by one binary tree over $A_u$, shared by the
whole row (\cref{fig:tree}): it ends in a single vertex $t_u$, and every
$x_{u,v}$ is computed from $t_u$. Then the pure player stores just
$t_1,\dots,t_K$ and computes any column from them in one round. All the
work of the source blocks has been ``summarized'' by $K$ checkpoints, and
the $\Theta(K^2)$ lower bound collapses. The same happens on the output
side: if the $K$ outputs of column $v$ are all computed from one vertex
$c_v$, storing $c_1,\dots,c_K$ makes every challenge cheap.

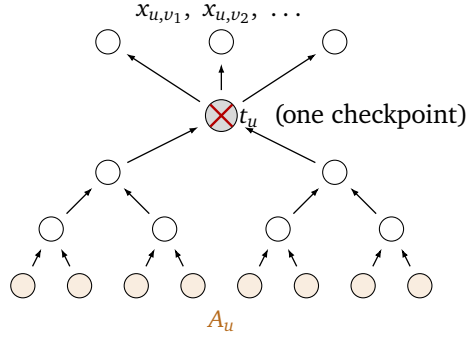
\begin{figure}[htbp]
\centering
\begin{tikzpicture}[x=0.75cm,y=0.75cm,font=\small,
    tv/.style={circle, draw, fill=white, inner sep=0pt, minimum size=3.2mm},
    te/.style={-{Latex[length=1.2mm,width=0.8mm]}, line width=0.5pt, shorten <=1.2mm, shorten >=1.2mm}]
  \foreach \i in {0,...,7} \node[tv, fill=buy!14] (s\i) at (\i,0) {};
  \foreach \i in {0,...,3} {
    \node[tv] (m\i) at (2*\i+0.5,1) {};
    \pgfmathtruncatemacro{\a}{2*\i}\pgfmathtruncatemacro{\b}{2*\i+1}
    \draw[te] (s\a) -- (m\i); \draw[te] (s\b) -- (m\i);
  }
  \node[tv] (n0) at (1.5,2) {}; \node[tv] (n1) at (5.5,2) {};
  \draw[te] (m0) -- (n0); \draw[te] (m1) -- (n0);
  \draw[te] (m2) -- (n1); \draw[te] (m3) -- (n1);
  \node[tv, fill=gray!30, minimum size=4.2mm] (t) at (3.5,3) {};
  \draw[te] (n0) -- (t); \draw[te] (n1) -- (t);
  \node[cross out, draw=red!70!black, line width=0.9pt, minimum size=2.6mm] at (t) {};
  \node[right=3pt] at (t) {$t_u$ \ (one checkpoint)};
  \foreach \j/\xx in {1/1.5, 2/3.5, 3/5.5} {
    \node[tv] (x\j) at (\xx,4.3) {};
    \draw[te] (t) -- (x\j);
  }
  \node[above] at (3.5,4.45) {$x_{u,v_1},\ x_{u,v_2},\ \dots$};
  \node[below, buy!85!black] at (3.5,-0.25) {$A_u$};
\end{tikzpicture}

\caption[A binary tree over a source block]{A binary tree over
$A_u$: one stored vertex $t_u$ replaces the whole block.}
\label{fig:tree}
\end{figure}

So we need a network in which a few stored vertices \emph{cannot}
intercept many input--output connections. That is exactly maximal
ST-robustness (\cref{def:st}), and the butterfly $\Bk$ of
\cref{prop:butterfly} has it with indegree two and depth $O(\log K)$. We
need it on \emph{both} sides: between the source blocks and the rows, and
between the columns and the outputs.

\subsection{The construction}
\label{sec:construction}

Let $K$ be a power of two and $G=([K],E)$ as before.

\begin{definition}[The DAG $\Gk$]
\label{def:gk}
\leavevmode
\begin{enumerate}[label=\textbf{\arabic*.}, leftmargin=*, itemsep=3pt]
  \item \textbf{Row connectors.} For every $u\in[K]$ a fresh copy $R_u$ of
        $\Bk$. Its input terminals are the sources
        $A_u = \{a_{u,1},\dots,a_{u,K}\}$, its output terminals are
        $r_{u,1},\dots,r_{u,K}$.
  \item \textbf{Edge vertices.} For every oriented edge $(u,v)$ a vertex
        $x_{u,v}$ with the single parent $r_{u,v}$:
        \[
          r_{u,v}\longrightarrow x_{u,v} .
        \]
  \item \textbf{Column connectors.} For every $v\in[K]$ a fresh copy $C_v$ of
        $\Bk$. Its $K$ input terminals are indexed by $u\in[K]$: for
        $u\in N(v)$ the input terminal \emph{is} the edge vertex $x_{u,v}$;
        for $u\notin N(v)$ it is a new source $\hat x_{u,v}$ (a dummy: the
        connector has $K$ inputs, but column $v$ has only $d_v$ labels). Its
        output terminals are the outputs $y_{v,1},\dots,y_{v,K}$.
\end{enumerate}
All copies are vertex-disjoint. The challenge set is again
$\Qk = \{y_{v,s} : (v,s)\in[K]^2\}$; we call its elements \emph{challenge
outputs}, to tell them apart from the outputs $r_{u,v}$ of the row
connectors.
\end{definition}

So the route from a source to a challenge is now
\[
  a\in A_u \ \rightsquigarrow\ r_{u,v} \ \longrightarrow\ x_{u,v} \ \rightsquigarrow\ y_{v,s},
\]
where $\rightsquigarrow$ is a path inside a connector ($R_u$ resp.\ $C_v$;
see \cref{fig:gk}).
The vertices $r_{u,v}$ and $x_{u,v}$ are kept distinct on purpose: $r_{u,v}$
belongs to the row connector $R_u$, $x_{u,v}$ to the column connector $C_v$,
so every vertex of $\Gk$ belongs to exactly one connector. This will let us
split the pure player's cache between connectors without counting anything
twice. (Outputs $r_{u,v}$ with $v\notin N(u)$ simply have no child.)

\begin{figure}[htbp]
\centering
\begin{tikzpicture}[x=1cm,y=1cm,font=\small,
    conn/.style={draw, rounded corners=3pt, fill=gray!8, minimum height=0.95cm},
    term/.style={draw, fill=white, rounded corners=1.5pt, inner sep=2pt, minimum height=0.42cm},
    arr/.style={-{Latex[length=1.6mm,width=1.1mm]}, line width=0.65pt}]
  % row connectors
  \foreach \j/\yy in {1/3.0, 2/1.6, 3/0.2} {
    \node[conn, minimum width=4.2cm] (R\j) at (2.3,\yy) {};
    \node[term, fill=buy!14, draw=buy!85!black] (A\j) at (0.75,\yy) {$A_{u_\j}$};
    \node[font=\scriptsize, gray] at (2.3,\yy+0.24) {$R_{u_\j}\cong\Bk$};
    \node[term] (r\j) at (3.85,\yy) {$r_{u_\j,v}$};
    \draw[gray, ->] (A\j.east) -- (r\j.west);
  }
  \node[gray] at (2.3,-0.75) {$\vdots$};
  % column connector
  \node[conn, minimum width=3.6cm, minimum height=5.0cm] (C) at (8.0,1.2) {};
  \node[font=\scriptsize, gray] at (8.3,3.5) {$C_v\cong\Bk$};
  \node[term, fill=ask!12, draw=ask!85!black] (y) at (9.25,1.4) {$y_{v,s}$};
  \foreach \j/\yy/\dy in {1/3.0/0.14, 2/1.6/0.05, 3/0.2/-0.05} {
    \node[term] (x\j) at (6.75,\yy) {$x_{u_\j,v}$};
    \draw[arr] (r\j) -- (x\j);
    \draw[gray, ->] (x\j.east) to[out=0, in=180] ($(y.west)+(0,\dy)$);
  }
  \node[term, dashed] (hx) at (6.75,-0.55) {$\hat x_{u,v}$};
  \draw[gray, ->, dashed] (hx.east) to[out=0, in=180] ($(y.west)+(0,-0.14)$);
  \node[font=\scriptsize, align=left, gray, anchor=north] at (6.75,-1.4) {dummy sources, $u\notin N(v)$};
  \node[font=\scriptsize, gray, align=center] at (9.25,0.75) {one of $K$\\ outputs};
\end{tikzpicture}

\caption[The DAG at indegree two around a challenge]{$\Gk$ around
a challenge $y_{v,s}$: a row connector for each neighbour $u\in N(v)$, the
edges $r_{u,v}\to x_{u,v}$ (black), and the column connector $C_v$. Gray
arrows stand for paths inside a connector.}
\label{fig:gk}
\end{figure}
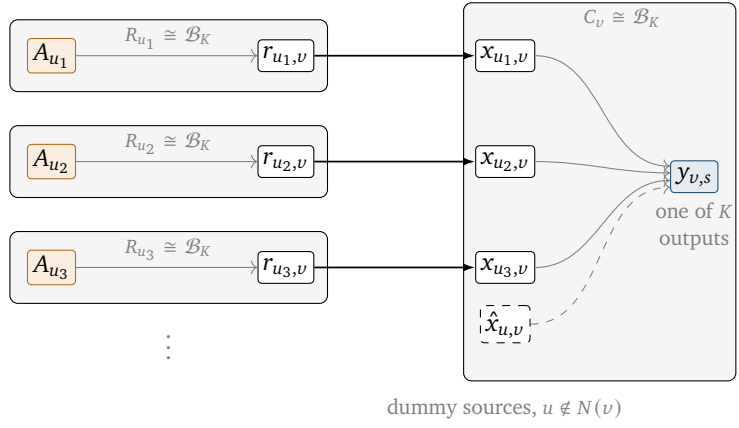

\begin{proposition}[Size and indegree]
\label{prop:size}
$\Gk$ has maximum indegree two, depth $O(\log K)$, and
\[
  N := |V(\Gk)| = 2K\,|V(\Bk)| = 2(2\log_2 K+1)K^2 = \Theta(K^2\log K).
\]
\end{proposition}

\begin{proof}
The $2K$ connectors are disjoint and contain every vertex (each $x_{u,v}$
and each dummy is an input terminal of exactly one column connector). Inside
a connector the indegree is at most two by \cref{prop:butterfly}; an edge
vertex $x_{u,v}$ has the single parent $r_{u,v}$; sources and dummies have
none. A source-to-output route crosses two connectors and one extra edge.
\end{proof}

The dictionary changes only in one entry: ``buying row $u$'' now means
evaluating the row connector $R_u$, layer by layer --- $K$ pebbles per
layer for $O(\log K)$ rounds, instead of $K$ pebbles for one round. The
graph got bigger by a $\log K$ factor, and so did the price of a row.

\subsection{Pure player}
\label{sec:pure-two}

The proof is the proof of \cref{lem:pure}, with one change: cached vertices
inside a connector do not cut single paths, they spoil some
\emph{terminals}, and ST-robustness tells us that $q$ cached vertices in a
connector spoil, together, at most $q$ of its inputs and $q$ of its
outputs. Summed over all connectors, at most $2p$ terminals are spoiled;
this is why the loss doubles from $p$ to $2p$.

\begin{lemma}[Pure lower bound on $\Gk$]
\label{lem:pure-two}
Let $P$ be any pure prepared state on $\Gk$, $p=|P|$. For a uniform
challenge $y_{v,s}\in\Qk$, every pure strategy satisfies
\[
  \E[\Cans] \ \ge\ p + (Kd-2p)_+ \ \ge\ \tfrac12 Kd .
\]
\end{lemma}

\begin{proof}
\emph{Step 1: good terminals.} For every connector $M\in\{R_1,\dots,R_K,
C_1,\dots,C_K\}$ let $q_M = |P\cap V(M)|$. The connectors are disjoint, so
$\sum_M q_M\le p$. Apply maximal ST-robustness (\cref{def:st}) to $M$ with
$D = P\cap V(M)$: we get a set of \emph{good inputs} and a set of
\emph{good outputs} of $M$, each missing at most $q_M$ terminals, none of
them in $P$, and with a path avoiding $P$ from every good input to every
good output. (If $q_M\ge K$, take both sets empty.) We fix this choice
depending only on $P$, before the challenge is drawn. Altogether at most
$2p$ terminals are \emph{not} good.

\emph{Step 2: certified sources.} For the challenge $y_{v,s}$, the
candidates are again the $Kd_v$ sources of the blocks $A_u$, $u\in N(v)$.
Call a candidate $a\in A_u$ \emph{certified} if all four terminals on its
route are good:
\begin{equation}
\label{eq:four-terminals}
  \underbrace{a}_{\text{\certmark{1} input of } R_u},\qquad
  \underbrace{r_{u,v}}_{\text{\certmark{2} output of } R_u},\qquad
  \underbrace{x_{u,v}}_{\text{\certmark{3} input of } C_v},\qquad
  \underbrace{y_{v,s}}_{\text{\certmark{4} output of } C_v}.
\end{equation}
For a certified $a$, glue a $P$-avoiding path $a\rightsquigarrow r_{u,v}$
in $R_u$, the edge $r_{u,v}\to x_{u,v}$ (both ends good, so outside $P$),
and a $P$-avoiding path $x_{u,v}\rightsquigarrow y_{v,s}$ in $C_v$. This is a
path from $a$ to the challenge avoiding $P$. Different $u$ have disjoint
blocks $A_u$, so the certified sources are distinct, and
\cref{lem:ancestors} gives, exactly as in \eqref{eq:pathwise},
\[
  \Cans \ \ge\ p + |U(v,s)| ,
\]
where $U(v,s)$ is the set of certified sources. (If $y_{v,s}\in P$ it is
not a good output, so $U(v,s)=\varnothing$ and the bound is just
$\Cans\ge p$.)

\Cref{fig:certified} shows one certified source. Its four terminals
\certmark{1}--\certmark{4} correspond to the four rows of the table in
Step~3.

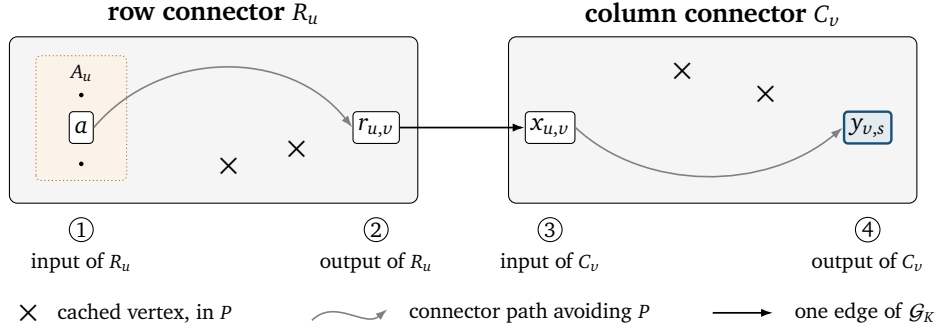
\begin{figure}[htbp]
\centering
\begin{tikzpicture}[x=1cm,y=1cm,font=\small,
    conn/.style={draw, rounded corners=3pt, fill=gray!8},
    term/.style={draw, fill=white, rounded corners=1.5pt, inner sep=2pt, minimum height=0.42cm},
    arr/.style={-{Latex[length=1.6mm,width=1.1mm]}, line width=0.65pt},
    cpath/.style={-{Latex[length=1.6mm,width=1.1mm]}, gray, line width=0.6pt},
    cached/.style={font=\large}]
  \draw[conn] (0,0) rectangle (5.4,2.2);
  \draw[conn] (6.6,0) rectangle (12.0,2.2);
  \node[font=\small\bfseries] at (2.7,2.5) {row connector $R_u$};
  \node[font=\small\bfseries] at (9.3,2.5) {column connector $C_v$};
  \draw[buy!85!black, densely dotted, fill=buy!10, rounded corners=2pt] (0.35,0.3) rectangle (1.55,1.95);
  \node[font=\scriptsize] at (0.95,1.72) {$A_u$};
  \fill (0.95,1.43) circle (0.03);
  \fill (0.95,0.52) circle (0.03);
  \node[term] (a) at (0.95,1.0) {$a$};
  \node[term] (r) at (4.85,1.0) {$r_{u,v}$};
  \node[term] (x) at (7.15,1.0) {$x_{u,v}$};
  \node[term, fill=ask!12, draw=ask!85!black, line width=0.9pt] (y) at (11.35,1.0) {$y_{v,s}$};
  \node[cached] at (2.9,0.5) {$\times$};
  \node[cached] at (3.8,0.72) {$\times$};
  \node[cached] at (8.9,1.75) {$\times$};
  \node[cached] at (10.0,1.45) {$\times$};
  \draw[cpath] (a.east) .. controls (2.0,2.05) and (3.7,2.05) .. (r.west);
  \draw[arr] (r) -- (x);
  \draw[cpath] (x.east) .. controls (8.4,0.15) and (10.0,0.15) .. (y.west);
  \foreach \p/\n/\t in {0.95/1/{input of $R_u$}, 4.85/2/{output of $R_u$},
                         7.15/3/{input of $C_v$}, 11.35/4/{output of $C_v$}} {
    \node at (\p,-0.35) {\certmark{\n}};
    \node[font=\scriptsize] at (\p,-0.78) {\t};
  }
  \node[font=\large] at (0.25,-1.45) {$\times$};
  \node[anchor=west, font=\scriptsize] at (0.5,-1.45) {cached vertex, in $P$};
  \draw[cpath] (4.0,-1.55) .. controls (4.3,-1.25) and (4.7,-1.65) .. (5.0,-1.4);
  \node[anchor=west, font=\scriptsize] at (5.15,-1.45) {connector path avoiding $P$};
  \draw[arr] (9.3,-1.45) -- (10.1,-1.45);
  \node[anchor=west, font=\scriptsize] at (10.25,-1.45) {one edge of $\Gk$};
\end{tikzpicture}

\caption[A certified source]{A certified source $a$ for the
challenge $y_{v,s}$. All four marked terminals are good, so a path in $R_u$,
the edge $r_{u,v}\to x_{u,v}$ and a path in $C_v$ join into a path from $a$
to the challenge that misses the cache (crosses). By \cref{lem:ancestors},
$a$ must then appear in a paid configuration.}
\label{fig:certified}
\end{figure}

\emph{Step 3: averaging.} A candidate loses its certificate only if one of
its four terminals is not good. How many certificates can one bad terminal
destroy, on average over the uniform $(v,s)$? (Again, primes mark the
indices of the bad terminal.)

\begin{center}
\small
\begin{tabular}{@{}llll@{}}
  \toprule
  bad terminal & candidates it kills & when? & expected loss \\
  \midrule
  \certmark{1} row input $a\in A_u$ & $1$ & $v\in N(u)$: prob.\ $d_u/K$ & $\le1$ \\
  \certmark{2} row output $r_{u,v'}$ & $K$ (all of $A_u$) & column is $v'$: prob.\ $1/K$ & $1$ \\
  \certmark{3} column input $x_{u,v'}$ & $K$ (all of $A_u$) & column is $v'$: prob.\ $1/K$ & $1$ \\
  \certmark{4} column output $y_{v',s'}$ & $Kd_{v'}$ (all) & challenge is $y_{v',s'}$: prob.\ $1/K^2$ & $d_{v'}/K\le1$ \\
  \bottomrule
\end{tabular}
\end{center}

(Unused row outputs and dummy column inputs lie on no candidate's route,
so they kill nothing.) There are at most
$2p$ bad terminals, each destroying at most one certificate in expectation,
so $\E[|U(v,s)|]\ge Kd - 2p$, and also $\ge (Kd-2p)_+$. This proves the
first inequality. For the second: if $p\ge Kd/2$ the term $p$ alone is
enough; otherwise $p + Kd - 2p = Kd - p > Kd/2$.
\end{proof}

\begin{intuition}[Same table, one row per terminal.]
Compare with the table in the proof of \cref{lem:pure}. There, a cached
vertex sat \emph{on} a path; here, the $q$ cached vertices inside a
connector cost, together, at most $q$ bad inputs and $q$ bad outputs of
that connector. The rows of the
table are the same kinds of objects --- share units, labels, answers --- now
appearing as terminals. ST-robustness is exactly what rules out the binary
tree's ``one checkpoint kills a whole block''. Note that the lower bound
does not care about the depth of the connectors; only the XOR side does.
\end{intuition}

The bound is tight up to the constant: caching all $K^2$ challenge outputs
costs $K^2$
by \eqref{eq:trivial}.

\subsection{Entangled player}
\label{sec:entangled-two}

The prepared state is the same as before, now built from the labels of
$\Gk$: $\Sx$ has one XOR pebble $\mathcal S_j = \{x_{a,b},x_{b,a} :
\{a,b\}\in M_j\}$ per matching, weight $m$. It is prepared by evaluating all
row connectors, computing all edge vertices, and replacing them by the
XORs. Define $H_v$ and $D_v$ by the same formulas as in
\cref{sec:entangled}. \Cref{lem:decode} depends only on $G$, the matchings
and which edge vertices are pebbled, so it holds on $\Gk$ word for word.

For a challenge $y_{v,s}$ the strategy is as before, but every ``one round''
becomes ``a connector evaluated layer by layer''. The XORs are kept all the
time.
\begin{enumerate}[label=\textbf{\arabic*.}, itemsep=3pt]
  \item \textbf{Rows.} Pebble the source blocks $A_h$, $h\in H_v$, and
        evaluate the row connectors $R_h$ in parallel, layer by layer,
        keeping only the current layer. From their outputs pebble the
        helper edge vertices $D_v$, and drop everything else.
  \item \textbf{Decode.} By \cref{lem:decode} the column
        $\{x_{u,v}:u\in N(v)\}$ is in the closure. Put pure pebbles on it in
        one move and drop $D_v$. Storing the column lets us drop the
        helpers before $C_v$ is evaluated, so they are not held during the
        $O(\log K)$ rounds of Column.
  \item \textbf{Column.} Pebble the dummy inputs $\hat x_{u,v}$ (sources);
        now all $K$ inputs of $C_v$ are pebbled. Evaluate $C_v$ layer by
        layer; in the last layer pebble only $y_{v,s}$. The answer is
        delivered.
  \item \textbf{Restore.} Remove $y_{v,s}$; we are back in $\Sx$.
\end{enumerate}

Every configuration is charged once in the table; the one holding the
helpers $D_v$, in which Rows ends, appears under Rows.

\begin{center}
\small
\begin{tabular}{@{}lll@{}}
  \toprule
  stage & stored besides $\Sx$ & weight per round $\times$ rounds \\
  \midrule
  Rows & sources $A(H_v)$, layers of $R_h$, then $D_v$ & $(m + O(K|H_v|))\times O(\log K)$ \\
  Decode & the column, $d_v$ pure pebbles & $(m + d_v)\times 1$ \\
  Column & layers of $C_v$ & $(m + O(K))\times O(\log K)$ \\
  Restore & the answer $y_{v,s}$ & $(m+1)\times 1$ \\
  \bottomrule
\end{tabular}
\end{center}

\begin{lemma}[Entangled upper bound on $\Gk$]
\label{lem:xor-two}
The strategy above is legal, answers every challenge and returns to exactly
$\Sx$, with
\[
  \Cret(v,s) = O\bigl(\log K\cdot(m + K|H_v| + K)\bigr),
  \qquad
  \E[\Cret] = O\bigl(\log K\cdot(m + K(K-d) + K)\bigr).
\]
\end{lemma}

\begin{proof}
\emph{Legality.} Sources and dummies can be pebbled at any time. A connector
layer is computed from the previous layer, which is pebbled. Each helper
vertex $x_{h,z}\in D_v$ has its parent $r_{h,z}$ in the last layer of $R_h$.
In Decode, every $x_{u,v}$, $u\in N(v)$, is in the closure by
\cref{lem:decode}, so it is available and may get a pure pebble, while the
helpers are dropped in the same move (\cref{def:entangled}). The
XOR pebbles are never changed and no other XOR pebble is placed, and
Restore is a removal. The
challenge enters the closure only in the last move of Column (before that,
no pebble and no XOR contains a challenge output), so the answer time is as
claimed.

\emph{Cost.} Read off the table: each connector has $K$ vertices per layer
and depth $O(\log K)$ (\cref{prop:butterfly}); $|D_v|\le K|H_v|$ and
$d_v\le K$. The XORs contribute $m$ in each of the $O(\log K)$ rounds.
Averaging uses $\E_v|H_v| = K-d$.
\end{proof}

\subsection{The separation at indegree two}
\label{sec:main}

We can now put together the pure lower bound (\cref{lem:pure-two}), the
XOR strategy (\cref{lem:xor-two}) and the graphs of \cref{prop:ams}. The
result is the main theorem of the paper.

\begin{theorem}[Main theorem]
\label{thm:main}
Let $\eta=0.029$ and let $G$ be a graph from \cref{prop:ams} on $K=64^\nu$
vertices. The DAG $\Gk$ has maximum indegree two and $N=\Theta(K^2\log K)$
vertices, and for a uniformly random challenge
\[
  \CPC_{\mathrm{pure}}(\Gk,\Qk) = \Theta\!\left(\frac{N}{\log N}\right),
  \qquad
  \CPC^{\mathrm{ret}}_{\mathrm{ent}}(\Gk,\Qk)
  = O\!\left(N^{1-\eta}(\log N)^{\eta}\right).
\]
So the entangled player still wins by a polynomial factor
$\Omega\bigl(N^{\eta}/(\log N)^{1+\eta}\bigr)$.
The entangled bound is achieved by a common-arity strategy of arity
$2\mu$.
\end{theorem}

\begin{proof}
By \cref{prop:size}, $N=\Theta(K^2\log K)$, so $\log N = \Theta(\log K)$
and $K^2 = \Theta(N/\log N)$. Pure: \cref{lem:pure-two} and
\eqref{eq:ams-consequences} give $\ge\frac12Kd=\Omega(K^2)$, and
\eqref{eq:trivial} gives $\le K^2$. Entangled: \cref{lem:xor-two} and
\eqref{eq:ams-consequences} give $O(K^{2-2\eta}\log K)$, and
$K^{2-2\eta}\log K = \Theta\bigl(N^{1-\eta}(\log N)^{\eta}\bigr)$. Divide.
\end{proof}

By \cref{rem:ackkpt}, this XOR strategy is also a strategy in the entangled
pebbling model of Alwen et al.~\cite{AlwenEtAl2016}, with the same cost;
so the gap of \cref{thm:main} persists when the entangled player is given
the full move rules of their model.

\begin{intuition}[Where did the $\log$ go?]
At high indegree the pure cost was $\Theta(N)$; now it is
$\Theta(N/\log N)$. Nothing got cheaper --- the pure player still pays
$\Theta(K^2)$ --- but the graph grew by a $\log K$ factor (the connector
layers), and those layers are not charged by our lower bound, which counts
only sources. On the other side the XOR player now pays $\log K$ times more,
because every row and column is evaluated over $O(\log K)$ rounds while he
holds his XORs. Both effects are logarithmic, so the polynomial gap
$N^{\eta}$ survives, up to logs.
\end{intuition}

\begin{center}
\small
\begin{tabular}{@{}lccc@{}}
  \toprule
  & memory game & $\Hk$ (indegree $K$) & $\Gk$ (indegree $2$) \\
  \midrule
  size & --- & $\Theta(K^2)$ & $\Theta(K^2\log K)$ \\
  pure (answer only) & $\ge Kd$ & $=Kd$ & $\ge\frac12Kd$ \\
  entangled (answer + return) & $m + K(K-d)$ & $\le 4m + 2K(K-d)+1$ & $O(\log K\,(m + K(K-d) + K))$ \\
  \midrule
  with AMS graphs: pure & $\Theta(K^2)$ & $\Theta(N)$ & $\Theta(N/\log N)$ \\
  with AMS graphs: entangled & $O(K^{2-2\eta})$ & $O(N^{1-\eta})$ & $O(N^{1-\eta}(\log N)^{\eta})$ \\
  \bottomrule
\end{tabular}\\*[6pt]
\begin{minipage}{0.9\linewidth}
In each DAG column, $N$ is the number of vertices of that DAG. The game
column lists $\TMC$; the DAG columns list the pure cost of answering and
the entangled cost of answering and returning.
\end{minipage}
\end{center}

% ======================================================================
\section{Repeated challenges}
\label{sec:repeated}

So far there was one challenge, and the prepared state was free. In a
data-dependent memory-hard function or a proof of space, challenges come
one after another, and the state has to be built first. Charging the
preparation once and answering $t$ fresh challenges in a row turns the
single-challenge separation into a separation of total costs, by a routine
amortization: the XOR state is built once, at cost $I$, and then reused.
If $L$ is the one-challenge pure lower bound of \cref{lem:pure-two} and $U$
the bound on the expected cost of answering and returning in
\cref{lem:xor-two}, then $t$ challenges cost the pure player at least $tL$
and the XOR player at most $I+tU$. Two properties of the one-shot proofs
make this work: the pure bound holds whatever state the earlier phases left
behind, and the XOR strategy returns to the same state. For large $t$ the
one-time cost $I$ becomes negligible, and what remains is the
per-challenge gap. \Cref{app:dmhf} writes out the data-dependent
memory-hard function behind the repeated game defined next.

\begin{definition}[The repeated game]
\label{def:repeated}
Let $\mathcal D$ and $\mathcal Q$ be as in \cref{def:service}, and fix
$t\ge1$. The player starts from the empty configuration and plays $t$
phases. In phase $i=1,\dots,t$:
\begin{enumerate}[label=\arabic*), itemsep=2pt]
  \item a challenge $c_i$ is drawn uniformly from $\mathcal Q$,
        independently of everything before, and revealed;
  \item the player makes legal moves
        $\mathcal P^{(i)}_0,\mathcal P^{(i)}_1,\dots$, where
        $\mathcal P^{(i)}_0$ is the configuration that ended phase $i-1$
        (empty for $i=1$); the phase ends at the first $\tau_i\ge1$ with
        $c_i\in\cl(\mathcal P^{(i)}_{\tau_i})$;
  \item the phase costs $\sum_{j<\tau_i}\wt(\mathcal P^{(i)}_j)$, as in
        \eqref{eq:cans}.
\end{enumerate}
$\CPC^{\mathrm{game}}_{\mathrm{pure}}(\mathcal D,\mathcal Q,t)$ and
$\CPC^{\mathrm{game}}_{\mathrm{ent}}(\mathcal D,\mathcal Q,t)$ are the
smallest expected total costs of pure, resp.\ entangled, strategies.
\end{definition}

This is \cref{def:service} played $t$ times in a row, with nothing free.
Since $c_1$ is revealed before the first move, building a state happens
inside phase~$1$ and is charged, and there are no moves between phases.
The condition $\tau_i\ge1$ matters: without it, a player who has cached
every output would answer all later challenges for free. The challenges
are supplied from outside, as in the randomized-query view
of~\cite[Section~2.2]{DER2026}; \cref{app:dmhf} shows where they come from
in a function.

\begin{theorem}[Repeated challenges]
\label{thm:repeated}
Let $G$ be a graph from \cref{prop:ams} and $t\ge1$. Then
\begin{gather*}
  \tfrac12\,t\,Kd \;\le\; \CPC^{\mathrm{game}}_{\mathrm{pure}}(\Gk,\Qk,t)
  \;\le\; O\bigl(K^2\log K + t\,K^2\bigr),\\
  \CPC^{\mathrm{game}}_{\mathrm{ent}}(\Gk,\Qk,t)
  \;=\; O\bigl(K^2\log K + t\,K^{2-2\eta}\log K\bigr).
\end{gather*}
In particular the pure cost is $\Theta(tK^2)$ for $t\ge\log_2K$. The
entangled bound is achieved by a common-arity strategy of arity $2\mu$
(\cref{def:common-arity}).
\end{theorem}

\begin{proof}
Write $c_i=y_{v_i,s_i}$.

\emph{Pure, lower bound.} Fix $i$ and condition on everything that
happened before $c_i$ is drawn. The configuration entering phase $i$ is
then fixed, $c_i$ is uniform on $\Qk$ and independent of it, and the phase
makes at least one move. So phase $i$ is an instance of \cref{def:service},
and by \cref{lem:pure-two} its conditional expected cost is at least
$\frac12Kd$, whatever the earlier phases left behind. Take expectations
and sum over $i$.

\emph{Pure, upper bound.} In phase~$1$ pebble all source blocks $A_u$,
then evaluate all row connectors, then pebble all edge vertices together
with the dummies, then evaluate all column connectors, layer by layer and
keeping only the current layer: $O(\log K)$ moves of $O(K^2)$ pebbles
each. No challenge output is pebbled before the last move, which pebbles
all $K^2$ of them at once; so phase~$1$ ends there, at cost $O(K^2\log K)$, in the
configuration $\Qk$. In every later phase one move that keeps $\Qk$
answers, at cost $K^2$.

\emph{Entangled, preparation.} Phase~$1$ starts after $c_1$ is revealed;
the strategy ignores $c_1$ at first and prepares $\Sx$ as in
\cref{sec:entangled-two}: pebble all source blocks, evaluate all row
connectors, pebble all edge vertices, and in one more move place the $m$
XOR pebbles and drop everything else. This costs $I=O(K^2\log K)$; no
pebble placed is a challenge output and no $\mathcal S_j$ contains one, so
phase~$1$
does not end early.

\emph{Entangled, serving the phases.} Then serve $c_1$ with Rows, Decode
and Column; phase~$1$ ends when Column delivers the answer. Every later
phase starts in $\Sx$ plus a pure pebble on the previous answer. Its first
move is the Restore of that answer, which leaves exactly $\Sx$; the closure
of $\Sx$ is empty (no pure pebbles, and every XOR pebble has at least two
elements), so the phase does not end there, even if the same challenge is
drawn twice in a row. Then Rows, Decode and Column serve the new challenge.

\emph{Entangled, accounting.} Let $\Cans(v,s)$ be the cost of answering
$y_{v,s}$ from $\Sx$ with this strategy (\cref{def:service}). Phase~$1$
costs $I+\Cans(v_1,s_1)$, and every later phase costs
$(m+1)+\Cans(v_i,s_i)$, where $m+1$ pays for Restore. Since
$\Cret(v,s)=\Cans(v,s)+m+1$ in \cref{lem:xor-two}, the total cost is
\[
  I+\sum_{i=1}^{t}\Cans(v_i,s_i)+(t-1)(m+1)\ \le\ I+\sum_{i=1}^{t}\Cret(v_i,s_i),
\]
whose expectation is $O(K^2\log K+tK^{2-2\eta}\log K)$ by
\cref{lem:xor-two} and \eqref{eq:ams-consequences}. The only XOR pebbles
ever stored are $\mathcal S_1,\dots,\mathcal S_m$, of size $2\mu$, so the
strategy is common-arity.
\end{proof}

\begin{corollary}[Repeated challenges, in terms of $n$]
\label{cor:repeated-n}
Let $G$ be a graph from \cref{prop:ams}, let $t=\Theta(N)$ (for example
$t=K^2\log_2K$), and let $n=N+2t+1$ be the number of vertices of the full
dependency graph of the function of \cref{app:dmhf}, which has maximum
indegree two (\cref{prop:full-size}); for $t=K^2\log_2K$,
$n=(6\log_2K+2)K^2+1$. Then
\[
  \CPC^{\mathrm{game}}_{\mathrm{pure}}(\Gk,\Qk,t)=\Theta\!\left(\frac{n^2}{\log n}\right),
  \qquad
  \CPC^{\mathrm{game}}_{\mathrm{ent}}(\Gk,\Qk,t)=O\bigl(n^{2-\eta}(\log n)^{\eta}\bigr),
\]
a gap of $\Omega\bigl(n^{\eta}/(\log n)^{1+\eta}\bigr)$. These are costs in
the repeated game; \cref{rem:scope} explains how they relate to evaluating
the function.
\end{corollary}

\begin{proof}
Here $n=\Theta(N)=\Theta(K^2\log K)$ by \cref{prop:size}, so
$\log n=\Theta(\log K)$, $t=\Theta(n)$ and $K^2=\Theta(n/\log n)$. As
$t\ge\log_2K$ for large $K$, \cref{thm:repeated} gives pure
cost $\Theta(tK^2)=\Theta(n^2/\log n)$. Entangled: $K^2\log K=O(n)$ and
$tK^{2-2\eta}\log K=\Theta\bigl(n\cdot(n/\log n)^{1-\eta}\log n\bigr)
=\Theta\bigl(n^{2-\eta}(\log n)^{\eta}\bigr)$. Divide.
\end{proof}

Taking $t=\Theta(N)$ pays for the preparation (its cost $O(K^2\log K)$ is
negligible next to $t$ times the cost of one challenge) and makes the line
about as long as $\Gk$ is large, so that $n=\Theta(N)$. The factor $1/\log n$ is only a change of
units: per challenge the pure player pays $\Theta(K^2)$, as in
\cref{thm:main}. For $t\gg N$ the line dominates, and the pure cost
$\Theta(tK^2)$ is $\Theta(nN/\log N)$ rather than $\Theta(n^2/\log n)$.

% ======================================================================
\section{Conclusion and open questions}
\label{sec:conclusion}

Pure pebbling can be wrong by a polynomial factor: the way stored
information is represented can change the asymptotic cost of memory, so a
pure-pebbling lower bound needs an additional argument before it says
anything about strategies that store encoded information. What should one
do about it, and how far does the gap go? We close with the questions we
find most pressing.

\paragraph{Which guarantees survive.} Not every security proof rests on
pebbling alone. For data-independent memory-hard functions evaluated from
scratch, pebbling lower bounds transfer to every algorithm in the parallel
random oracle
model~\cite{AlwenSerbinenko2015}, so they also hold against algorithms that
store encoded information.
Scrypt is maximally memory-hard against every algorithm~\cite{Scrypt2017}.
The proof of space of Dziembowski et al.~\cite{DFKP2015}, first analysed
only against adversaries that store a subset of the honest prover's data,
was later proved secure against arbitrary adversaries in the random oracle
model~\cite{Pietrzak2018}, and proofs of space built on function inversion
instead of pebbling are secure against any stored advice by
design~\cite{AACKPR2017}. The weak spot is the data-dependent world. There
no general reduction from algorithms to pebbling is
known~\cite{BlockiHolman2022,BlockiHolman2025}, and the best bounds against
arbitrary algorithms are polynomially weaker than what pebbling
predicts~\cite{BlockiHolman2025}. Our result, like that of de~Rezende,
Engstr\"om and Reyzin~\cite{DER2026}, shows that such a gap need not be an
artifact of the proofs: it can be real.

\paragraph{How to defend.} For a designer, the lesson is that a
pure-pebbling bound is not enough. One wants a proof against arbitrary
algorithms, or at least a bound in a model that allows encoded storage,
such as entangled pebbling, in which Alwen et al.\ proved their bound for
scrypt~\cite{AlwenEtAl2016}. We know no general technique for lower bounds
in entangled pebbling, and no simple property of a graph that makes XORs
useless. The memory game points to one warning sign: a query that needs
one value from each of many groups, where a group can only be recomputed as
a whole, like the rows of the game. Then a pure player wastes most of every
group he recomputes, and stored XORs can make the wasted values useful.

\paragraph{Is XOR the end?} Our adversary is weak: it stores XORs of a
single kind. Richer encodings, such as linear combinations over larger
fields or nonlinear functions of labels, might gain more. There is a reason
to expect so. In network coding, which stands to routing roughly as
entangled pebbling stands to pebbling, coding beats routing, and nonlinear
codes provably beat linear ones~\cite{DFZ2005}. Whether the same happens
for memory, and how large the gap between pebbling and arbitrary algorithms
can be, is open; we prove no lower bound for any entangled strategy.

\paragraph{Real constructions.} We do not claim that any deployed
memory-hard function or proof of space is broken; our graphs are built for
the attack. Whether XORs give a polynomial advantage on the graphs used in
practice is perhaps the most practical question we leave open. A smaller
one: the exponent $\eta$ comes from the graphs of Alon, Moitra and Sudakov,
in which the number of missing edges and the number of induced matchings
are both bounded by about $K^{1.94}$ (\cref{app:ams}). Any family with both
numbers polynomially smaller improves $\eta$ through the same argument.

Pebbling remains a powerful abstraction, but a pebble is a model of memory,
not memory itself.

% ======================================================================
\appendix

\section{The AMS graphs: parameters and explicitness}
\label{app:ams}

This appendix proves \cref{prop:ams}. We need a power-of-two number
of vertices (for the butterflies of \cref{prop:butterfly} and for the
challenge map $\chi$ of \cref{app:dmhf}) and simultaneous bounds on the
number of missing edges and of matchings. We also record that the graph is
regular and that the partition has classes of equal size. Neither enters
any bound: the proofs use only the average degree $d$ and work for every
partition into induced matchings. Equal size only makes our XOR strategies
common-arity (\cref{def:common-arity}).

In the code-based construction of~\cite{AMS2013}, vertices are words of
length $\nu$ over an alphabet of size $C$, so $K = C^\nu$, and two words are
adjacent when they differ in more than $\nu - d_\nu$ coordinates. The
parameter $\rho = \lim d_\nu/\nu$ trades the two quantities: the exponent
$e_C(\rho)$ below controls the missing edges (the rows the entangled player
buys), and $f_C(\rho)$ controls the number of matchings (the XORs he
stores). We need one $\rho$ making $\max\{e_C(\rho),f_C(\rho)\}$ small.
$H_2(z) = -z\log_2 z-(1-z)\log_2(1-z)$ is the binary entropy.

\begin{theorem}[Code-based AMS graphs~{\cite[Theorem~2.20]{AMS2013}}]
\label{thm:ams}
Fix an integer $C\ge6$ and $\frac{2}{C-1}\le\rho<\frac12$. There are
integers $d_\nu$ with $d_\nu/\nu\to\rho$, regular graphs $G_\nu$ on
$K=C^\nu$ vertices, and partitions of their edges into $m_\nu$ induced
matchings, with
\[
  M_{\mathrm{miss}}(G_\nu)\le K^{e_C(\rho)+o(1)},
  \qquad
  m_\nu\le K^{f_C(\rho)+o(1)},
\]
where
\[
  e_C(\rho)=1+\frac{H_2(\rho)+(1-\rho)\log_2(C-1)}{\log_2C},
  \qquad
  f_C(\rho)=2-\frac{1-H_2(\rho)}{\log_2C}.
\]
Moreover, all matchings of the partition have one common size
$\mu=K^{2-f_C(\rho)+o(1)}$.
\end{theorem}

\begin{proof}
The two bounds are~\cite[Theorem~2.20]{AMS2013}. In the construction the
edges are grouped into equivalence classes~\cite[Lemma~2.18]{AMS2013}, and
every class is an induced matching with exactly $2^{k_\nu-1}$ edges, where
$k_\nu=(1-H_2(\rho)+o(1))\nu$~\cite[Lemma~2.19]{AMS2013}. Hence
$2^{k_\nu-1} = K^{(1-H_2(\rho))/\log_2C+o(1)} = K^{2-f_C(\rho)+o(1)}$.
The classes partition the edges, so this is an equal-size partition.
Adjacency depends only on Hamming distance and coordinatewise alphabet
permutations act transitively, so the graph is regular.
\end{proof}

\paragraph{Balancing the exponents.} For $C=64$ the two exponents move in
opposite directions on the whole admissible range $2/63\le\rho<1/2$: since
$H_2'(\rho)=\log_2((1-\rho)/\rho)$,
\[
  e_{64}'(\rho)=\frac{\log_2((1-\rho)/\rho)-\log_2 63}{6}<0,
  \qquad
  f_{64}'(\rho)=\frac{\log_2((1-\rho)/\rho)}{6}>0
\]
(the first because $\rho>1/64$, the second because $\rho<1/2$). So the
maximum is minimized where they cross. In general $e_C(\rho)=f_C(\rho)$ is
equivalent to $(1-\rho)\log_2(C-1)=\log_2C-1$.

\begin{corollary}[Balanced parameters]
\label{cor:ams-balanced}
Let $C=64$ and $\rho_0=1-5/\log_2 63\approx0.1634991$. Then
$2/63<\rho_0<1/2$ and
\[
  e_{64}(\rho_0)=f_{64}(\rho_0)=:\lambda_0\approx1.9404359 .
\]
Consequently, for $\eta=0.029$ and all large enough $\nu$, the regular graph
on $K=64^\nu$ vertices has an equal-size induced-matching partition with
$m\le K^{2-2\eta}$ and $M_{\mathrm{miss}}\le K^{2-2\eta}$.
\end{corollary}

\begin{proof}
The crossing condition with $\log_2 64=6$ gives $\rho_0$; substituting gives
$\lambda_0$. Since $\lambda_0<2-2\eta=1.942$, the $o(1)$ terms of
\cref{thm:ams} are eventually below the gap. $K=64^\nu$ is a power of two.
\end{proof}

Any $0<\eta<(2-\lambda_0)/2\approx0.0297821$ works; we take $\eta=0.029$.
This is a convenient constant for these estimates, not a claim of
optimality over all constructions. On these graphs the common
matching size is $\mu=K^{2-\lambda_0+o(1)}$ (\cref{thm:ams} at $\rho_0$),
so our XOR strategies have arity $2\mu=K^{2-\lambda_0+o(1)}$, with
$2-\lambda_0\approx0.0596$.

\paragraph{Deterministic construction.} The graph and its matchings must be
known to define the XOR state. AMS call a binary linear code \emph{proper}
if it contains the all-ones vector~\cite[Definition~2.8]{AMS2013}.

\begin{proposition}[Explicitness]
\label{prop:ams-explicit}
For the parameters of \cref{cor:ams-balanced}, the graph $G_\nu$ and its
equal-size induced-matching partition can be listed deterministically in
time polynomial in $K=64^\nu$.
\end{proposition}

\begin{proof}
Compute a rational $\rho_\nu$ with $|\rho_\nu-\rho_0|\le\nu^{-2}$ and set
$d_\nu=\lceil\nu\rho_\nu\rceil$; then $d_\nu/\nu\to\rho_0$ and
$2/63\le d_\nu/\nu<1/2$ for large $\nu$. Put
$V_\nu=\sum_{r\le d_\nu}\binom{\nu}{r}$ and
$k_\nu=\nu-\lceil\log_2 V_\nu\rceil\ge1$.

\emph{The code.} Start with $\mathcal C=\operatorname{span}\{\mathbf 1\}$.
While $\dim\mathcal C=j<k_\nu$, adjoin the lexicographically first vector
at Hamming distance more than $d_\nu$ from every codeword. It exists, since
the forbidden vectors lie in $2^j$ Hamming balls of total size at most
$2^{k_\nu-1}V_\nu\le2^{\nu-1}<2^\nu$. The enlarged code is the old one plus
one coset, so its minimum distance stays above $d_\nu$, and it still
contains $\mathbf 1$. By the standard estimate
$\log_2V_\nu=H_2(\rho_0)\nu+o(\nu)$, so $k_\nu=(1-H_2(\rho_0)+o(1))\nu$.
This takes $2^{O(\nu)}$ time. For $0\le r<d_\nu$, puncturing the last $r$
coordinates gives proper codes $\mathcal C_{\nu-r}$ of the same dimension
and distance above $d_\nu-r$~\cite[Claim~2.10]{AMS2013}.

\emph{The graph and the matchings.} Enumerate the graph on $[64]^\nu$. Take
the lexicographically first edge $\{a,b\}$ not yet assigned; if $a,b$ agree
in $r$ coordinates, apply every codeword $x\in\mathcal C_{\nu-r}$ as an AMS
coordinate flip on the unequal coordinates (the $x$-flip
of~\cite[Definition~2.13]{AMS2013}: keep the $r$ coordinates where $a$ and
$b$ agree, and on the other $\nu-r$ swap the letters of $a$ and $b$ exactly
where $x$ has a one). This lists its whole class,
which has exactly $2^{k_\nu-1}$ edges and is an induced
matching~\cite[Lemma~2.19]{AMS2013}. Mark it and repeat. The total time is
$K^2\,2^{k_\nu}\operatorname{poly}(\nu)=K^{O(1)}$.
\end{proof}

The connectors are also listed in time polynomial in $K$, so the DAG $\Gk$
and the XOR state are fully explicit.

% ======================================================================
\section{The function behind the repeated game}
\label{app:dmhf}

This appendix writes out the function behind the repeated game of
\cref{sec:repeated}: where its challenges come from, why they are uniform,
and what $n$ in \cref{cor:repeated-n} counts. In a data-dependent
memory-hard function the challenges come from the labels themselves: a line
of new vertices is appended to the static DAG, and in each phase one vertex
of the line selects an output of the DAG, which the next vertex reads
(\cref{fig:line}) ---
as in scrypt~\cite{Scrypt2017}.

\begin{definition}[DAG-based data-dependent memory-hard function]
\label{def:dmhf}
A \emph{DAG-based dMHF} is given by a static DAG $\mathcal D$, a challenge
set $\mathcal Q\subseteq V(\mathcal D)$, a number $t\ge1$ of phases and a
selection map $\sigma:\{0,1\}^w\to\mathcal Q$. Its \emph{full dependency
graph} on input $x$ is $\mathcal D$ together with a \emph{line} of $2t+1$
new vertices $\ell_0,g_1,\ell_1,\dots,g_t,\ell_t$, where $\ell_0$ is a
source. In phase $i$ the \emph{seed} $g_i$ has the single parent
$\ell_{i-1}$ and selects the challenge $c_i=\sigma(X_{g_i})$; the
\emph{response} $\ell_i$ has exactly the two parents $g_i$ and $c_i$. The
output is $X_{\ell_t}$.
\end{definition}

The name follows the literature and describes the shape of the
construction; we prove no memory-hardness for evaluating it
(\cref{rem:scope}).

On $(\Gk,\Qk)$ we take $\sigma=\chi$, defined as follows. Write $K=2^k$ (so
$k=6\nu$ for the graphs of \cref{prop:ams}) and assume $w\ge2k$. The map
$\chi:\{0,1\}^w\to\Qk$ sends a string to $y_{v,s}$, where its first $k$
bits are the binary expansion of $v-1$ and the next $k$ bits that of
$s-1$. All labels are computed by \eqref{eq:label}; on the line this
reads, for $i=1,\dots,t$,
\begin{equation}
\label{eq:line}
\begin{aligned}
  X_{\ell_0} &= \mathsf H\bigl(x,\mathsf{id}(\ell_0)\bigr), &
  X_{g_i} &= \mathsf H\bigl(x,\mathsf{id}(g_i),X_{\ell_{i-1}}\bigr),\\
  c_i &= \chi(X_{g_i}) = y_{v_i,s_i}, &
  X_{\ell_i} &= \mathsf H\bigl(x,\mathsf{id}(\ell_i),X_{g_i},X_{c_i}\bigr),
\end{aligned}
\end{equation}
and the output is $X_{\ell_t}$. Only the edges $c_i\to\ell_i$ depend on
the input $x$. The pair $(v_i,s_i)$ must be uniform on $[K]^2$, not merely
each coordinate: if $s_i$ were a function of $v_i$ (even a bijection, so
that both coordinates are uniform), a pure player could cache $K$ outputs and answer every challenge for about $K$ ---
exactly what the $K$ outputs per column were introduced to prevent
(\cref{sec:dag-def}). $\chi$ achieves this because the $2k$ bits it reads
are independent.

\begin{intuition}[A hash value is as good as a die.]
The game rolls a die for $c_i$; the function computes $c_i$ from
$X_{g_i}$. The two have the same distribution because $X_{g_i}$ is a fresh
random-oracle answer (distinct identifiers make its query new). For $i>1$
that query contains $X_{\ell_{i-1}}$, and $X_{\ell_{i-1}}$ needs the answer
to $c_{i-1}$; so in the computation $c_i$ comes only after $c_{i-1}$ is
answered, and then it is as unpredictable as a fresh roll. What the game idealizes is only the moment when $c_i$ becomes known,
and the cost of the line (\cref{rem:scope}). The line makes the challenges
unpredictable; the memory sits in $\Gk$.
\end{intuition}

\begin{proposition}[The function on $\Gk$]
\label{prop:full-size}
For every $t\ge1$ and every input $x$:
\begin{enumerate}[label=(\alph*), itemsep=2pt]
  \item the full dependency graph is acyclic, has maximum indegree two, and
        has exactly $n=N+2t+1=2(2\log_2K+1)K^2+2t+1$ vertices;
  \item if $\mathsf H$ is a random oracle and $w\ge2\log_2K$, the
        challenges $c_1,\dots,c_t$ are independent and uniform on $\Qk$,
        and independent of the labels of $\Gk$.
\end{enumerate}
\end{proposition}

\begin{proof}
(a) $\Gk$ has $N$ vertices and indegree at most two (\cref{prop:size}), and
no vertex of $\Gk$ gets a parent on the line. The line adds $2t+1$
vertices: the source $\ell_0$, and for each $i$ the seed $g_i$ (one
parent, $\ell_{i-1}$) and the response $\ell_i$ (two parents, $g_i$ and
$c_i$). Every line edge goes from $\Gk$ into the line or forward along it,
so there is no cycle.

(b) Fix $i$ and condition on the values of $\mathsf H$ at all queries whose
identifier is not $\mathsf{id}(g_i)$. They determine the labels of $\Gk$
and $X_{\ell_0},X_{g_1},\dots,X_{\ell_{i-1}}$, hence $c_1,\dots,c_{i-1}$
and the query that defines $X_{g_i}$. Since tuples are encoded
unambiguously and identifiers are distinct (\cref{sec:labels}), that query
is not among the fixed ones, so its answer is uniform on $\{0,1\}^w$ and
independent of the conditioning. As $w\ge2k$, $c_i=\chi(X_{g_i})$ is
uniform on $\Qk$ given the labels of $\Gk$ and $c_1,\dots,c_{i-1}$.
\end{proof}

\begin{figure}[htbp]
\centering
\begin{tikzpicture}[x=1cm,y=1cm,font=\small,
    lv/.style={draw, fill=white, rounded corners=1.5pt, inner sep=2.5pt, minimum height=0.5cm},
    arr/.style={-{Latex[length=1.6mm,width=1.1mm]}, line width=0.6pt}]
  \node[gray] at (-1.1,0) {$\cdots$};
  \node[lv] (lp) at (0,0)   {$\ell_{i-1}$};
  \node[lv] (g)  at (2.4,0) {$g_i$};
  \node[lv] (l)  at (5.4,0) {$\ell_i$};
  \node[lv] (gn) at (7.6,0) {$g_{i+1}$};
  \node[gray] at (8.7,0) {$\cdots$};
  \node[lv, fill=ask!12, draw=ask!85!black] (y) at (5.4,1.6) {$c_i=y_{v_i,s_i}$};
  \draw[arr] (lp) -- (g);
  \draw[arr] (g) -- (l);
  \draw[arr] (l) -- (gn);
  \draw[arr, ask!80!black] (y) -- (l);
  \draw[arr, dashed, gray] (g.north) to[out=90, in=180]
    node[pos=0.35, left=3pt, font=\scriptsize] {$\chi(X_{g_i})$} (y.west);
  \node[anchor=west, gray, font=\scriptsize] at (6.55,1.6) {an output of $C_{v_i}$ in $\Gk$};
\end{tikzpicture}

\caption[One phase of the line]{One phase of the line. Solid
arrows are edges of the full dependency graph. The dashed arrow is
\emph{not} an edge: the label of $g_i$ only selects which output becomes the
second parent of $\ell_i$.}
\label{fig:line}
\end{figure}
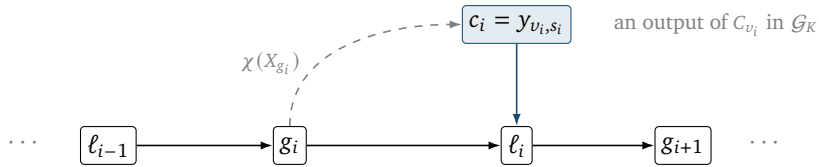

\begin{remark}[What the game measures]
\label{rem:scope}
\Cref{thm:repeated,cor:repeated-n} are statements about the repeated game
on $(\Gk,\Qk)$; $n$ is the size of the full dependency graph of the
function behind it. The game differs from evaluating that function in two
ways. It \emph{draws} each challenge at the start of its phase, while the
function \emph{computes} it from $X_{g_i}$; by \cref{prop:full-size}(b)
the distribution is the same, but the game also fixes the moment at which
the challenge becomes known. And it charges only configurations of $\Gk$:
the $2t+1$ line vertices count in $n$, but they are neither pebbled nor
charged. We do not claim these bounds for pebbling the full dependency
graph with each challenge revealed only when its seed is computed, and, as
everywhere in this paper, we claim nothing for general algorithms in the
parallel random oracle model.
\end{remark}

\newpage

\section*{Acknowledgments}
I am grateful to Krzysztof Pietrzak for his mentorship throughout this
project, for suggesting the problem, and for insightful discussions
that helped shape the single-challenge perspective and improve the
organization and presentation of this paper.

I used ChatGPT and Claude to assist with drafting and editing the
exposition, LaTeX formatting, and consistency checks. I reviewed the
resulting text and take full responsibility for the mathematical
claims, proofs, and final presentation.
% ======================================================================


\begin{thebibliography}{99}

\bibitem{AACKPR2017}
H.~Abusalah, J.~Alwen, B.~Cohen, D.~Khilko, K.~Pietrzak, and L.~Reyzin.
Beyond Hellman's time-memory trade-offs with applications to proofs of
space. In \emph{Advances in Cryptology--ASIACRYPT 2017}, 2017.

\bibitem{AlwenSerbinenko2015}
J.~Alwen and V.~Serbinenko. High parallel complexity graphs and memory-hard
functions. In \emph{Proceedings of STOC 2015}, pages 595--603, 2015.

\bibitem{AlwenBlocki2016}
J.~Alwen and J.~Blocki. Efficiently computing data-independent memory-hard
functions. In \emph{Advances in Cryptology--CRYPTO 2016}, Part~II, pages
241--271, 2016.

\bibitem{ABP2017}
J.~Alwen, J.~Blocki, and K.~Pietrzak. Depth-robust graphs and their
cumulative memory complexity. In \emph{Advances in Cryptology--EUROCRYPT
2017}, pages 3--32, 2017.

\bibitem{AlwenEtAl2016}
J.~Alwen, B.~Chen, C.~Kamath, V.~Kolmogorov, K.~Pietrzak, and S.~Tessaro.
On the complexity of scrypt and proofs of space in the parallel random
oracle model. In \emph{Advances in Cryptology--EUROCRYPT 2016}, pages
358--387, 2016.

\bibitem{Scrypt2017}
J.~Alwen, B.~Chen, K.~Pietrzak, L.~Reyzin, and S.~Tessaro. Scrypt is
maximally memory-hard. In \emph{Advances in Cryptology--EUROCRYPT 2017},
pages 33--62, 2017.

\bibitem{AMS2013}
N.~Alon, A.~Moitra, and B.~Sudakov. Nearly complete graphs decomposable into
large induced matchings and their applications. \emph{Journal of the
European Mathematical Society}, 15(5):1575--1596, 2013.

\bibitem{Benes1964}
V.~E.~Bene\v{s}. Optimal rearrangeable multistage connecting networks.
\emph{Bell System Technical Journal}, 43(4):1641--1656, 1964.

\bibitem{BlockiCinkoske2021}
J.~Blocki and M.~Cinkoske. A new connection between node and edge depth
robust graphs. In \emph{12th Innovations in Theoretical Computer Science
Conference (ITCS 2021)}, LIPIcs 185, Article 64, 2021.

\bibitem{BlockiHolman2022}
J.~Blocki and B.~Holman. Sustained space and cumulative complexity
trade-offs for data-dependent memory-hard functions. In \emph{Advances in
Cryptology--CRYPTO 2022}, pages 222--251, 2022.

\bibitem{BlockiHolman2025}
J.~Blocki and B.~Holman. Data-dependent memory-hard functions: Sustained
space and cumulative complexity trade-offs in the parallel random oracle
model. In \emph{Advances in Cryptology--CRYPTO 2026}, LNCS 16805, pages
35--64, 2026.

\bibitem{DER2026}
S.~F. de~Rezende, D.~Engstr\"om, and L.~Reyzin. Separating the pebbling
model from the random oracle model. Cryptology ePrint Archive, Report
2026/1024, 2026.

\bibitem{DFZ2005}
R.~Dougherty, C.~Freiling, and K.~Zeger. Insufficiency of linear coding in
network information flow. \emph{IEEE Transactions on Information Theory},
51(8):2745--2759, 2005.

\bibitem{DFKP2015}
S.~Dziembowski, S.~Faust, V.~Kolmogorov, and K.~Pietrzak. Proofs of space.
In \emph{Advances in Cryptology--CRYPTO 2015}, pages 585--605, 2015.

\bibitem{Pietrzak2018}
K.~Pietrzak. Proofs of catalytic space. In \emph{10th Innovations in
Theoretical Computer Science Conference (ITCS 2019)}, 2019. Cryptology ePrint Archive, Report 2018/194.

\end{thebibliography}
\end{document}